\documentclass[acmsmall,screen,nonacm]{acmart}

\AtBeginDocument{%
  \providecommand\BibTeX{{%
    \normalfont B\kern-0.5em{\scshape i\kern-0.25em b}\kern-0.8em\TeX}}}

\setcopyright{rightsretained}
\acmPrice{}
\acmDOI{10.1145/3571236}
\acmYear{2023}
\copyrightyear{2023}
\acmSubmissionID{popl23main-p303-p}
\acmJournal{PACMPL}
\acmVolume{7}
\acmNumber{POPL}
\acmMonth{1}
\received{2022-07-07}
\received[accepted]{2022-11-07}

\usepackage{cleveref}
\usepackage{mathpartir}
\usepackage{caption}
\usepackage{subcaption}
\usepackage{amsmath}
\usepackage{hyperref}
\usepackage{fancyvrb}
\usepackage{booktabs}
\usepackage[shortlabels]{enumitem}
\usepackage{tikz}
\usetikzlibrary{decorations.pathmorphing,shapes,arrows}
\usepackage{rotating}
\usepackage{csquotes}
\usepackage{multicol}
\setlist[enumerate,1]{1.}
\usepackage{setspace}

\newlength{\codeindent}
\newcommand{\codeexample}{\footnotesize\setstretch{1.0}\addtolength{\jot}{-0.4em}}
\newcommand{\Infer}[3]{\ensuremath{\inferrule*[right={#1}]{#2}{#3}}}
\newcommand{\trm}[1]{\textrm{#1}}
\newcommand{\ttt}[1]{\texttt{#1}}
\newcommand{\tsc}[1]{\textsc{#1}}
\newcommand{\defunc}{~\rightsquigarrow~}

\newcommand{\transformsto}[1]{\xrightarrow{\ \ {#1}\ \ }}
\newcommand{\many}[1]{\underline{{{#1}}\vphantom{\dot{{#1}}}}}
\newcommand{\dbar}[1]{\many{\dot{#1}}}
\newcommand{\ddbar}[1]{\many{\ddot{#1}}}
\newcommand{\letx}[3]{\ttt{let(}~#1~\ttt{) =}~#2~\ttt{in}~#3}
\newcommand{\defx}[4]{\ttt{def}~#1\ttt{(}~#2~\ttt{)} \to \ttt{(}#3\ttt{) =}~#4}
\newcommand{\defxArrow}[5]{\ttt{def}~#1\ttt{(}~#2~\ttt{)} #4\to \ttt{(}#3\ttt{) =}~#5}

\newcommand{\R}{\mathbb{R}}
\newcommand{\Le}{\mathcal{L}}
\newcommand{\T}{\mathcal{T}}
\newcommand{\Pa}{\mathcal{U}}
\newcommand{\J}{\mathcal{J}}
\newcommand{\eps}{\varepsilon}
\newcommand{\ring}{\mathcal{R}}
\DeclareMathOperator{\dup}{dup}
\DeclareMathOperator{\drop}{drop}

\newcommand{\vocab}[1]{\emph{{#1}}}
\newcommand{\Lone}{Linear A}
\newcommand{\Lonestrict}{Linear B}

\newcommand{\parens}[1]{\ttt{(} {#1} \ttt{)}}
\newcommand{\tupOp}[1]{\otimes \ttt{(} #1 \ttt{)}}

\newcommand{\dotp}[2]{\langle {#1}, {#2}\rangle}
\newlength{\wdth}

\newcommand{\rl}[1]{\tsc{{#1}}}     %
\newcommand{\W}{\mathcal{W}}
\newcommand{\inargs}{L_{\trm{in}}}
\newcommand{\outargs}{L_{\trm{out}}}

\newcommand{\elval}{\mathsf}
\newcommand{\delval}[1]{\dot{\elval{#1}}}
\newcommand{\ddelval}[1]{\ddot{\elval{{#1}}}}
\newcommand{\melval}[1]{\many{\elval{{#1}}}}
\newcommand{\mdelval}[1]{\many{\delval{{#1}}}}
\newcommand{\mddelval}[1]{\many{\ddelval{{#1}}}}

\newcommand{\envNL}{\Gamma}         %
\newcommand{\envTan}{\dot{\Gamma}}  %
\newcommand{\envCot}{\ddot{\Gamma}} %
\newcommand{\envT}{\envCot}         %
\newcommand{\seqT}{\ddot{\Delta}}         %
\newcommand{\seqTan}{\dot{\Delta}}         %
\newcommand{\opT}[2]{\T(#1 \mid #2)}
\newcommand{\typet}[1]{#1}          %

\newcommand{\lin}[1]{{#1}\ttt{.lin}}
\newcommand{\nonlin}[1]{{#1}\ttt{.nonlin}}
\newcommand{\opP}[2]{\Pa(#2)}  %
\newcommand{\hole}{\square}

\newcommand{\jvp}[2]{\J({#1} \mid {#2})}
\newcommand{\typej}[1]{{#1}}    %

\newif\ifextended
\extendedfalse

\begin{document}



\title{You Only Linearize Once}
\subtitle{Tangents Transpose to Gradients}

\author{Alexey Radul}
\email{axch@google.com}
\affiliation{%
  \institution{Google Research}
  \country{USA}
}

\author{Adam Paszke}
\email{apaszke@google.com}
\affiliation{%
  \institution{Google Research}
  \country{Poland}
}

\author{Roy Frostig}
\email{frostig@google.com}
\affiliation{%
  \institution{Google Research}
  \country{USA}
}

\author{Matthew J. Johnson}
\email{mattjj@google.com}
\affiliation{%
  \institution{Google Research}
  \country{USA}
}

\author{Dougal Maclaurin}
\email{dougalm@google.com}
\affiliation{%
  \institution{Google Research}
  \country{USA}
}

\renewcommand{\shortauthors}{A. Radul, A. Paszke, R. Frostig, M. Johnson, D. Maclaurin}

\begin{abstract}
Automatic differentiation (AD) is conventionally understood as a family of distinct algorithms, rooted in two ``modes''---forward and reverse---which are typically presented (and implemented) separately.
Can there be only one?
Following up on the AD systems developed in the JAX and Dex projects, we formalize a decomposition of reverse-mode AD into (i) forward-mode AD followed by (ii) unzipping the linear and non-linear parts and then (iii) transposition of the linear part.

To that end, we define a (substructurally) linear type system that can prove a class of functions are (algebraically) linear.
Our main results are that forward-mode AD produces such linear functions, and that we can unzip and transpose any such linear function, conserving cost, size, and linearity.
Composing these three transformations recovers reverse-mode AD.
This decomposition also sheds light on checkpointing, which emerges naturally from a free choice in unzipping \ttt{let} expressions.
As a corollary, checkpointing techniques are applicable to general-purpose partial evaluation, not just AD.

We hope that our formalization will lead to a deeper understanding of automatic differentiation
and that it will simplify implementations,
by separating the concerns of differentiation proper from the concerns of gaining efficiency
(namely, separating the derivative computation from the act of running it backward).
\end{abstract}

\maketitle

\section{Introduction}
\label{sec:introduction}
Automatic differentiation (AD) powers not only deep machine learning, but also applications in broader numerical computing, from robotics to molecular dynamics to weather simulation.
For example, in deep learning the task of choosing how to set the billions of parameters of a neural network is almost always attacked with AD: a programmer writes a prediction function representing how the network maps parameters and example inputs to predicted outputs, as well as a scalar-valued loss function penalizing deviations from expected outputs to predicted ones.
Then to incrementally improve a given setting of the parameters, the programmer need only ask an AD system for the gradient of the loss with respect to the parameters on some example inputs and outputs, then update the parameters by adding a small negative multiple of that gradient.
Repeat a million times, et voil\`{a}, la singularit\'{e} est proche!

AD is traditionally organized into two \emph{modes}, forward-mode and reverse-mode.
Each mode has its advantages and applications.
But most systems only implement one mode, or if both are implemented, the implementations are almost entirely separate.

Do they need to be?
After all, reverse-mode AD is similar to forward-mode AD, except one computes the derivative part backward.
We take this description literally: Our goal is to separate differentiation proper from direction reversal,
and identify a clear boundary between them. 
The ``direction reversal'' is \vocab{transposition}, an interesting and potentially useful code transformation in its own right.
The boundary between differentiation and transposition consists of provably linear functions: linearity (proven by our type system) is a sufficient condition for transposability, and we show that automatic derivatives are provably linear.

This simplifies implementations of reverse-mode AD (see \Cref{fig:ad-full}).  The actual differentiation is the province of the widely known and relatively simple (and covariant!) forward-mode transformation.
Reverse mode is derived therefrom by a separate transformation that now must only confront linearity, not derivatives. 
The notion of linearity we formalize here serves as a simpler and clearer (and mechanically verifiable) input invariant for transposition than one of the form ``this program was generated by forward-mode AD, so it involves no unacceptable thing X.''

This same architecture also makes AD systems easier for a user to extend.  A user wishing to supply an efficient custom derivative for some function need only supply the forward (or reverse) derivative; as long as it is provably linear, the reverse (resp., forward) derivative can be derived automatically.

Reverse-mode automatic differentiation is already implemented in the proposed way in two systems we know of: JAX\footnote{The \texttt{jax.grad} function is exactly a composition of three functions that implement the steps we outline in this work.}~\citep{frostig2018compiling, jax2018github} and Dex~\citep{paszke2021dex}.  Our hope is to make the technique easier to maintain within these systems, by formalizing the provable linearity invariant; and to make it more widely accessible outside of them by describing it on a simple, stand-alone object language.
We also hope to pave the way for exposing transposition as a user-accessible higher-order function in its own right, by clarifying its input and output invariants.

More concretely, our contributions are that:
\begin{itemize}
\item We explicitly decompose reverse-mode automatic differentiation into forward mode (\Cref{sec:jvp}), a new unzipping transformation (\Cref{sec:partial}), and transposition (\Cref{sec:transposition}).
\item We introduce a new linearly-typed intermediate language for linear computations (\Cref{sec:language}), whose type system captures the provable linearity invariant. We call the language ``\Lone{}.''
\item We define and prove correctness for unzipping and transposition, and we prove that they preserve work, code size, and provable linearity.  Together with the known results on correctness and work- and size-preservation for forward mode, this implies correctness and work- and size-preservation of reverse mode when implemented this way.\footnote{In our cost model, transposition and unzipping preserve work exactly, and program size up to a constant factor.  Forward mode, as usual, only preserves work up to a constant factor.}
\end{itemize}

\begin{figure}
\centering
\begin{subfigure}[t]{0.08\textwidth}
    \centering
    \includegraphics[height=1.1in]{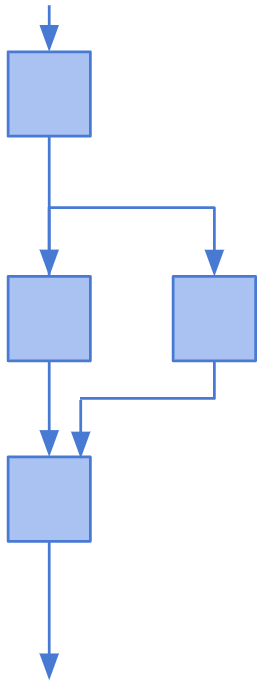}
    \caption{$f$}
    \label{fig:ad-1-orig}
\end{subfigure}
\hfill
\begin{subfigure}[t]{0.16\textwidth}
    \centering
    \includegraphics[height=1.1in]{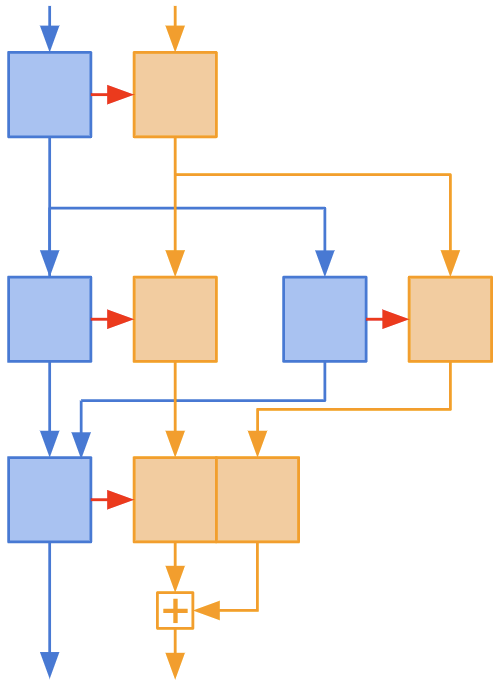}
    \caption{forward derivative}
    \label{fig:ad-2-jvp}
\end{subfigure}
\hfill
\begin{subfigure}[t]{0.25\textwidth}
    \centering
    \includegraphics[height=1.1in]{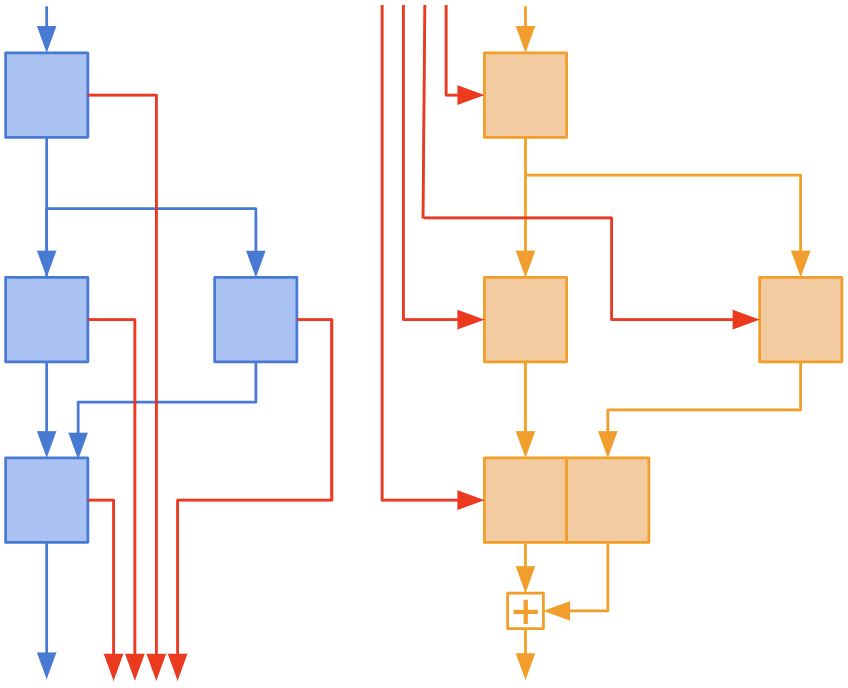}
    \caption{forward derivative, unzipped}
    \label{fig:ad-3-unzip}
\end{subfigure}
\hfill
\begin{subfigure}[t]{0.14\textwidth}
    \centering
    \includegraphics[height=1.1in]{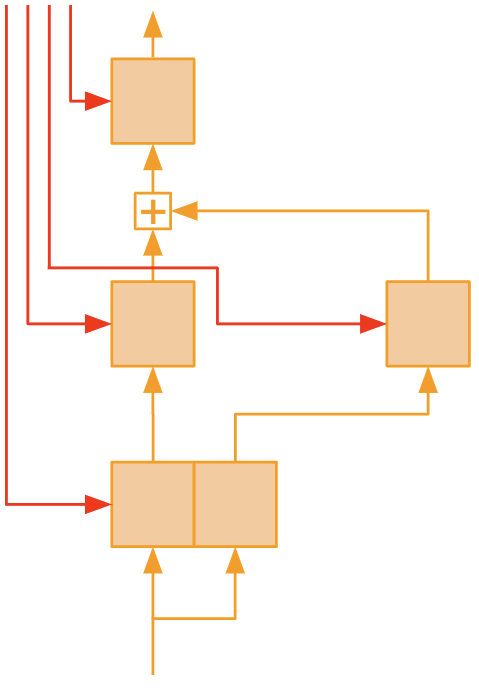}
    \caption{derivative, transposed}
    \label{fig:ad-4-transpose}
\end{subfigure}
\hfill
\begin{subfigure}[t]{0.17\textwidth}
    \centering
    \includegraphics[height=1.1in]{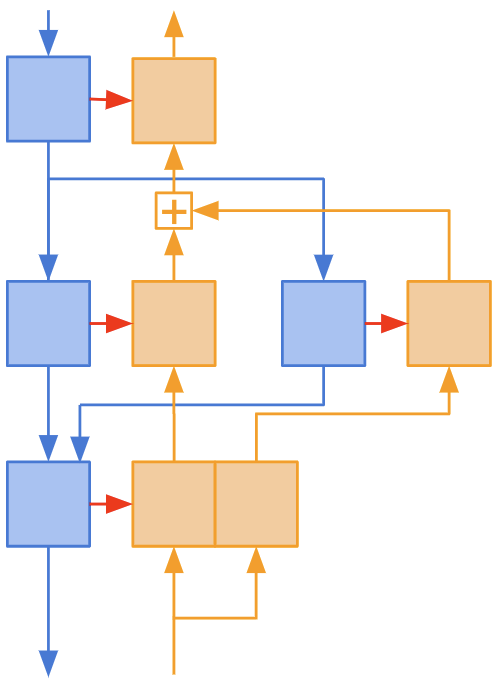}
    \caption{reverse derivative}
    \label{fig:ad-5-vjp}
\end{subfigure}
\caption{Reverse-mode automatic differentiation as forward differentiation, unzipping, and transposition.}
\label{fig:ad-full}
\end{figure}

\section{Preliminaries}

\subsection{Mathematical Differentiation Operations Corresponding to AD's Two Modes}
At its core, AD brings two particular differentiation operations from mathematical functions to programming language functions.
These two operations correspond to the forward and reverse modes.
The first operation, the \emph{Jacobian-vector product (JVP)}, underlies forward-mode and answers the question: if I perturb my input $x$ a bit in some direction $v$, how does my output change?
More precisely, given a sufficiently nice mathematical function $f : \mathbb{R}^n \to \mathbb{R}^m$, we can define the Jacobian at $x \in \mathbb{R}^n$ as the linear map  $\partial f(x) : \mathbb{R}^n \to \mathbb{R}^m$ such that
\begin{equation*}
f(x + v) = f(x) + \partial f(x)(v) + o(\| v \|),
\qquad
\forall v \in \mathbb{R}^n.
\end{equation*}
We then say a \emph{Jacobian-vector product for $f$} is the mapping
\begin{equation*}
(x, v) \mapsto (f(x), \, \partial f(x)(v)).
\end{equation*}
This definition is compositional, in that the JVP of $f \circ g$ can be expressed by applying the JVP of $f$ to the result of the JVP of $g$. It also allows for efficient evaluation, as we make precise in the sequel.

The second core mathematical operation represented in automatic differentiation is the \emph{vector-Jacobian product (VJP)}, which underlies reverse-mode.
It answers a more subtle question: given a linear function on small perturbations to a function's output, what is the corresponding linear function on small perturbations to the input?
That is, for some fixed $x$, given a vector $u \in \mathbb{R}^m$, find the vector $w \in \mathbb{R}^n$ such that
\begin{equation*}
\langle u, \, \partial f(x)(v) \rangle = \langle w, \, v \rangle,
\qquad
\forall v \in \mathbb{R}^n,
\end{equation*}
where $\langle \cdot , \, \cdot \rangle$ denotes the standard inner product.
This mapping from $u$ to $w$ is in fact linear, so we can define a new linear map $\partial f(x)^T : \mathbb{R}^m \to \mathbb{R}^n$ by
\begin{equation*}
\langle u, \, \partial f(x)(v) \rangle = \langle \partial f(x)^T(u), \, v \rangle.
\end{equation*}
We then say the \emph{vector-Jacobian product for $f$} is the mapping
\begin{equation*}
x \mapsto (f(x), \, u \mapsto \partial f(x)^T(u)).
\end{equation*}
This choice of definition is also compositional: the VJP of $f \circ g$ can be expressed by composing the VJPs of $f$ and $g$, though the linear functions must be composed in reverse order as $\partial g(x)^T \circ \partial f(x)^T$. Hence the name "reverse-mode"!

The VJP is the more useful operation for gradient-based optimization, like in training neural networks. That's because with one VJP we compute the direction in parameter space orthogonal to the loss's local level sets. That is, for a function $f: \mathbb{R}^n \to \mathbb{R}$, the gradient is simply $\nabla f(x) = \partial f(x)^T(1) \in \mathbb{R}^n$. The gradient can also be computed using forward-mode, but then it requires $n$ evaluations, losing the asymptotic efficiency of AD.

The implementation of reverse-mode is often considered much harder than that of forward-mode, with little shared code between the two.  Yet, the ``real work'' of differentiation is embedded in the linear functions $\partial f(x)$ and $\partial f(x)^T$, which are just transposes of each other.

\subsection{Linearity}

Linearity is the key to relating forward and reverse modes. The word ``linear'' actually has two relevant meanings when discussing a programming language for numerical computations:
\begin{enumerate}
    \item From substructural logic and linear type systems, a subprogram $P(x)$ is \vocab{substructurally linear} if it uses its argument $x$ exactly once.
    \item From linear algebra, a (mathematical) function $f(x)$ is \vocab{algebraically linear} if it is a vector space homomorphism, i.e., if $f(x + y) = f(x) + f(y)$ and $f(cx) = cf(x)$ for scalar $c$.
\end{enumerate}

These two notions are not named the same by accident.
Indeed, when is a polynomial, written as a sum of products of variables (no exponentiation), algebraically linear in $x$?
Exactly when every nonzero monomial term uses $x$ exactly once.
This is how, when we get to \Cref{thm:linearity}, we will be able to prove algebraic linearity using an appropriately designed substructural type system.

The Jacobian-vector product (JVP) map at a point is, by mathematical definition, linear algebraically, but it also turns out (as we will also show in \Cref{sec:jvp}) that its instantiation in AD is linear substructurally.
That substructural linearity will, in turn, allow us to mechanically transpose (\Cref{sec:transposition}) the program for computing $f$ into a (also substructurally linear) program that computes $f^\T$ with the same performance.  This then recovers reverse-mode AD, which computes the vector-Jacobian-product (VJP) map.

Part of the point of this exercise is that the transposition transformation we will define in \Cref{sec:transposition} applies to all substructurally linear functions.
Substructurally linear functions thus form an abstraction boundary between differentiation and transposition; and transposition (of user-defined functions) could then be exposed as a code transformation directly.

Note that there are algebraically linear functions that our type system will not accept as substructurally linear, and on which our transposition transformation would fail.
For example, the function $f(x) = (x * x) / x$ is linear in $x$, but it cannot be typed as linear in \Lone\ as written.
Moreover, attempting to transpose such a function as we do, by recurring on the subexpressions, is doomed to failure---the subexpression $\ldots / x$ is not linear in $x$ in isolation, so has no transpose.
Fortunately, as we prove in \Cref{sec:jvp}, automatic differentiation never produces such functions.

\section{Notation Reference}

By convention, we use an \textbf{over-dot} to name linear terms: $\dot v$.
The goal is to evoke the conventional physics notation for derivatives, since automatic differentiation is our main topic.
Formally, however, we give the dot itself no structured meaning.  So, the variable $\dot v$ is just a different variable from $v$; any relationship it may have to $v$ emerges from our program transformations, rather than from how we write it.

The \textbf{double-dot} $\ddot v$ connotes co-tangents (not double derivatives), i.e., linear terms that appear in a transposed function.  The nomenclature is from the reverse phase of reverse-mode AD.

The \textbf{underline} $\many v$ means ``zero or more of these elements''; we put it below rather than above to avoid clashing with dots: $\dbar v$.
When the same thing appears in different underlined expressions in the same context, they are parallel.
For instance, when a rule mentions both $\many v$ and $\many {v: \tau}$, those lists are the same length.

The \textbf{semicolon} $(x; y)$ separates non-linear entities (on the left) from linear ones (on the right).

The \textbf{angle brackets} $\langle x, y \rangle$ mean dot product when defining the meaning of transposition, and delimit transposition environments when explaining its implementation.

The \textbf{square brackets} $e[\many x; \many y]$ mean evaluating an expression with values $\many x$ and $\many y$ bound to its free non-linear and linear variables respectively.

\section{Language}
\label{sec:language}

We introduce \Lone, a model language of \vocab{indexed linear} computations.
The main idea in \Lone\ is that the syntax marks which values are supposed to be linear (substructurally and therefore algebraically) and which are not.
Non-linear values can be computed arbitrarily (and may happen to be algebraically linear), but may not depend on linear values.
Linear values, on the other hand, must be computed linearly from linear inputs, but may depend on non-linear values through scaling.
This leads to the indexed-linear pattern of data flow shown in \Cref{fig:dependencies}.

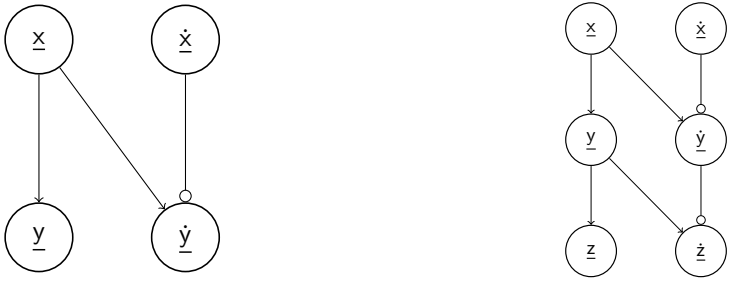
\begin{figure}
\centering
\begin{minipage}{0.49\columnwidth}
\centering
\scalebox{1}{%
\begin{tikzpicture}[
  y=-1cm,
  scale=1.3,
  decoration=snake,
  xx/.style={circle, draw=black, inner sep=0pt, minimum size=9mm, semithick},
  dd/.style={}
]
\node[xx] (xp) at (0,   0) {$\melval{x}$};
\node[xx] (xt) at (1.5, 0) {$\mdelval{x}$};
\node[xx] (yp) at (0,   2) {$\melval{y}$};
\node[xx] (yt) at (1.5, 2) {$\mdelval{y}$};

\draw [->] (xp) -- (yp) node[midway,left] {};
\draw [->] (xp) -- (yt) node[midway,above] {};
\draw [-o] (xt) to (yt);
\end{tikzpicture}
}%
\subcaption{Expression $e$ on variables $\melval{x}; \mdelval{x}$ with results $\melval{y}; \mdelval{y}$}
\label{fig:dependencies:block}
\end{minipage}
\hfill
\begin{minipage}{.49\columnwidth}
\centering
\scalebox{.75}{%
\begin{tikzpicture}[
  y=-1cm,
  scale=1.3,
  decoration=snake,
  xx/.style={circle, draw=black, inner sep=0pt, minimum size=9mm, semithick},
  dd/.style={}
]
\node[xx] (xp) at (0,   0) {$\melval{x}$};
\node[xx] (xt) at (1.5, 0) {$\mdelval{x}$};
\node[xx] (yp) at (0,   1.5) {$\melval{y}$};
\node[xx] (yt) at (1.5, 1.5) {$\mdelval{y}$};
\node[xx] (zp) at (0,   3) {$\melval{z}$};
\node[xx] (zt) at (1.5, 3) {$\mdelval{z}$};

\draw [->] (xp) -- (yp) node[midway,left] {};
\draw [->] (xp) -- (yt) node[midway,above] {};
\draw [-o] (xt) to (yt);
\draw [->] (yp) -- (zp) node[midway,left] {};
\draw [->] (yp) -- (zt) node[midway,above] {};
\draw [-o] (yt) to (zt);
\end{tikzpicture}
}%
\subcaption{Sequential composition of expressions $e_1, e_2$}
\label{fig:dependencies:stack}
\end{minipage}
\caption{Dataflow diagrams of
(a) a general \Lone\ expression $e$,
and (b) sequential composition of such expressions.
The non-linear results
of an expression (left) depend only only on its non-linear free variables,
whereas the linear results (right, dotted) depend
non-linearly on the non-linear free variables
and linearly on the linear free variables.  Composition preserves this (non-)dependence pattern.}
\label{fig:dependencies}
\end{figure}

For example, the function $g(x; y) = (x^2; x^2 y)$ is indexed linear.
The first result of $g$ depends non-linearly on $x$ and does not depend on $y$ at all; whereas the second result depends linearly on $y$, but the specific linear function of $y$ is indexed by $x$.
If we tag $x$ and the first result as ``non-linear'', and $y$ and the second result as ``linear'', $g$ will type-check as a valid function of \Lone.\footnote{It will also type-check if we tag both inputs and both outputs ``non-linear,'' but not with any other tagging.}

The \Lone\ language is not designed as a user-facing language, but rather as a sublanguage that a given instance of automatic differentiation operates on.
To wit, if one seeks to differentiate some $F$ at some point $x$, the value of $x$ determines the control flow choices in the implementation of $F$, and the program corresponding to the path actually taken (sometimes called a ``trace'') becomes straight-line.
Differentiating such straight-line trace programs is the core task of AD, and we restrict 
\Lone\ to be total, functional, and first-order in order to focus our attention on it.
While a more expressive object language would certainly be preferable, the technique of AD by explicit transposition has not been previously described even in this restricted setting; in \Cref{sec:future-work}, we will touch on some obstacles to formalizing transposition on a Turing-complete object language.

Many differentiable array languages are DSLs of the same expressive power as \Lone, embedded in a host language that provides Turing-complete metaprogramming facilities.
Furthermore, reverse-mode automatic differentiation of control flow constructs (including higher-order functions) cannot be performed without forming a dynamic trace of the computation, and the languages conventionally used in those traces have expressive power similar to \Lone{}.
For example, \citet{pearlmutter2008lambda} and \citet{wang2019shiftreset} represent the trace as a sequence of applications.
Similarly, \citet{krawiec2021provably} propose an AD system capable of differentiating a rich higher-order functional language, but the differentiation involves tracing an expression in a language almost identical to the linear fragment of \Lone.
Finally, the language used internally in the JAX project \citet{jax2018github}, which implements our presented method, is also very close to \Lone.
The main differences are that JAX has a significantly larger number of mathematical primitives, as well as some limited control-flow operators (conditionals, and loops with a statically known trip-count).
None of these are fundamental to the presented method and for that reason we elide them for simplicity.
JAX additionally supports array types with statically known shapes, which are isomorphic to product types in \Lone.
On the other hand, \Lone\ is slightly richer than JAX's internal language, as it allows nested let-bindings to have an arbitrary expression on the right hand side, not only immediate primitive applications.

Many details of \Lone\ are unconventional, so we include point-by-point rationales.
We focus on the syntax in \Cref{sec:l1-syntax}, then present and discuss the type system in \Cref{sec:l1-typing}, state our cost model in \Cref{sec:l1-cost}, and formally define our notion of indexed linearity in \Cref{sec:l1-linearity}, where we also prove that the type system enforces it.
Finally, \Cref{sec:linearity-erasure} defines the natural linearity erasure transformation on \Lone, which we will occasionally need later.

\subsection{Syntax}
\label{sec:l1-syntax}

\begin{figure}[t]\scriptsize
\begin{tabular}{r@{~~}r@{~~}ll}
  \multicolumn{3}{l}{Types} \\
  $\tau, \sigma, \pi$ & $::=$ & $\R$ & real scalar \\
  & $\mid$ & $\otimes \many \tau$         & tuple type \\
  \\
  \multicolumn{3}{l}{Values (extralinguistic ground constants)} \\
  $x, y, \dot x, \dot y$ & $::=$ & number in $\R$ & real scalar \\
  & $\mid$ & $\otimes \many x$            & tuple \\
  \\
  \multicolumn{3}{l}{$m,n$-ary Expressions} & $e$ is short-hand for $e^{1,0}$ and $\dot e$ for $e^{0,1}$ \\
  $e^{m,n}$ & $::=$
           & \ttt{(} $\many v; \dbar v$ $\ttt{)}$
             & Multi-value return; $|\many v| = m, |\dbar v| = n$ \\
  & $\mid$ & $\letx {\many {v:\tau}; \many {\dot v:\dot \tau}} {e^{o,p}} {e^{m,n}}$
             & multi-value let; $|\many v| = o, |\dbar v| = p$ \\
  & $\mid$ & $\letx {\otimes {\many w}:\otimes {\many \tau};} {v} {e^{m,n}}$
             & unpacking let for non-linear tuples; $|\many w| = |\many \tau|$  \\
  & $\mid$ & $\letx {;\otimes {\dbar w}:\otimes {\dbar \tau}} {\dot v} {e^{m,n}}$
             & unpacking let for linear tuples; $|\dbar w| = |\dbar \tau|$ \\
  & $\mid$ & $f^{o,p}_{m,n} \ttt{(} \many v; \dbar v \ttt{)}$
             & $o,p$-in, $m,n$-out function application; $|\many v| = o, |\dbar v| = p$ \\
  \\
  \multicolumn{3}{l}{Fixed-arity Expressions} \\
  $e^{1,0}$ & $::=$
           & $u, v, w, z$ & non-linear variable \\
  & $\mid$ & $l$ & literal; $l \in \R$ \\
  & $\mid$ & $\tupOp {\many v}$ & tuple constructor \\
  & $\mid$ & $\sin(v) \mid \cos(v) \mid \exp(v) \mid v_1 + v_2 \mid \ldots $ & primitives for non-linear fragment\\
  $e^{0,1}$ & $::=$
           & $\dot u, \dot v, \dot w, \dot z$ & linear variable \\
  & $\mid$ & $\dot 0_{\dot \tau}$ & linear zero of type $\dot \tau$ \\
  & $\mid$ & $\tupOp {\dbar v}$ & tuple constructor \\
  & $\mid$ & $\dot v_1 + \dot v_2$ & linear addition \\
  & $\mid$ & $v * \dot v$ & right-linear multiplication \\
  $e^{0,2}$ & $::=$
           & $\dup(\dot v)$ & explicit fan-out for linear values \\
  $e^{0,0}$ & $::=$
           & $\drop(e^{m,n})$ & explicit drop \\
  \\
  \multicolumn{3}{l}{Definitions} \\
  $d$ & $::=$ & $\defx {f^{m,n}_{o,p}} {\many {v: \tau}; \many{\dot v: \dot \tau}}
                       {\many \tau;\dbar \tau} {e^{o,p}}$ & $m,n$-in, $o,p$-out function; $|\many v| = m$, $|\dbar v| = n$, $|\many \tau| = o$, $|\dbar \tau| = p$ \\
  \\
  \multicolumn{3}{l}{Programs} \\
  $P$ & $::=$ & $\epsilon \mid d, P$ & list of function definitions \\
  \\
  \multicolumn{3}{l}{Environments} \\
  $\Gamma$ & $::=$ & $\epsilon \mid v : \tau, \Gamma$ & unordered type environment \\
  $\Sigma$ & $::=$ & $\epsilon \mid v \to w, \Sigma$ & variable association map (for $\J$ in \Cref{sec:jvp}) \\
  $\Delta$ & $::=$ & $\langle \rangle \mid \langle v : \tau, \Delta \rangle$ & ordered type environment (for $\T$ in \Cref{sec:transposition}) \\
  \\
  \multicolumn{3}{l}{Contexts (for $\Pa$ in \Cref{sec:partial})} \\
  $E$ & $::=$ & $\hole$ & null context \\
  & $\mid$ & $\letx {\many{v:\tau};} {e^{n,0}} {E}$ & non-linear let context, $|\many v| = n$ \\
  & $\mid$ & $\letx {\otimes \many w : \otimes \many \tau ;} {v} {E}$ & non-linear unpack context \\
  & $\mid$ & $E_1, E_2$ & composition of contexts, by hole-substitution
\end{tabular}
\caption{Syntax of \Lone{}.
Expressions are syntactically indexed by how many non-linear and linear results they return.
A function name enters scope after its definition---no recursive function calls.
The underline $\many v$ means ``zero or more of these elements''; we put it below rather than above to avoid clashing with dots: $\dbar v$.
}
\label{fig:linear-a}
\end{figure}

The syntax of \Lone\ appears in \Cref{fig:linear-a}.  We make \Lone\ a first-order language by only permitting function definitions at the top level.
We make \Lone\ total by disallowing recursion: function names are bound only after their definition, as given by the program order.
We also syntactically mark which values are linear vs not.
In our notation, this distinction shows up as returns, binders, and function application forms taking two lists of arguments (or formal parameters): the non-linear arguments before the semicolon and the linear ones after.

Our function $g(x;y) = (x^2;x^2y)$ from the previous section would be written:%
\begingroup\codeexample%
\begin{align*}
\defxArrow
  {& \ttt{g}^{1,1}_{1,1}}
  {x: \R; \dot y:\R}
  {\R; \R} {}
  {\\
&   \letx {a : \R ;} {x * x} {\\
&   \letx {; \dot b : \R} {a * \dot y} {\\
&   \parens{a; \dot b}
}}}
\end{align*}%
\endgroup
Note that $*$ (as well as $+$) are available as primitives to both linear and non-linear computations in \Lone, just with different typing rules.

\subsubsection{Tuples and Multiple Returns}

We distinguish between product types (written $\otimes$), which are first-class in \Lone, and multiple-value returns (written $\parens{\many v; \dbar v}$), which are the only construct permitted to mention both linear and non-linear values together.  An expression returning multiple values cannot be bound to one variable, but must instead be bound componentwise.
We enforce this in the syntax by indexing every expression $e^{m,n}$ by the number $m$ of non-linear results it returns and the number $n$ of linear results it returns.  As a short-hand, we write $e$ for a one-value non-linear expression, and $\dot e$ for a one-value linear expression.  Products, in contrast, may be bound to variables, but must contain only linear or only non-linear values.

Multiple-value returns are a somewhat unusual syntactic choice.  In our case, they are convenient, because our transposition transformation of \Cref{sec:transposition} will be running expressions ``backwards''.
Since an expression is free to reference more than one variable from its environment, its transpose must be able to return more than one result; and it's more elegant to do so directly in the syntax.

\subsubsection{A-Normal Form}
The \Lone\ syntax enforces administrative normal form \citep{sabry1992anf} on expressions.
Where needed, A-normal form can be enforced by standard techniques;
we ask the reader to indulge us when we write some transformation results without explicitly introducing all the intermediate variables that our syntax technically requires.
We do this for clarity and brevity.
Under the same aegis, we also only allow primitives to be unary or binary, and all primitives except $\dup$ and $\drop$ return exactly one result.
We likewise occasionally abuse notation and write $\parens{;e,\many v}$ where $e$ may return multiple results.
The intended meaning is an expression that executes $e$ and returns all of its results, and also the variables $\many v$.
This can be arranged within the syntax of \Lone\ by allocating fresh names to briefly hold the results of $e$.

A-normal form does make our work preservation results easier.  To wit, it's possible and reasonable to implement a differentiate-unzip-transpose AD system that operates directly on compound expressions.
Such a system would, however, need to be careful to introduce temporaries to avoid duplicating code expressions (and therefore work); for example in rule \rl{JPrimMul} in \Cref{fig:jvp}.  A-normal form saves us from that because the intermediate names are already present.

\subsubsection{Asymmetric Linear Multiply}
Linear computations in \Lone\ model real multiplication as linear in the right-hand argument and non-linear in the left-hand argument.
(This materializes in rule \rl{TypeLinMul} in \Cref{fig:l1-typerules} checking the left argument in the non-linear environment and the right argument in the linear environment.)
Which is of course immaterial, because real multiplication commutes, but we intentionally pick one to keep the linearity of variables syntactically apparent.
Of course, multiplication is actually \emph{multi}-linear, but we leave explicitly modeling this as an open avenue for future research.

\begin{figure}\scriptsize
\renewcommand{\arraystretch}{3}
\begin{tabular}{cc}
\multicolumn{2}{c}{
\framebox{$\Gamma; \dot \Gamma \vdash e^{n,m} : \parens{\many \tau; \dbar \tau},\ |\many \tau| = n,\ |\dbar \tau| = m$}} \\
\multicolumn{2}{c}{
$\Infer{TypeRet}{~~}
  {\many{v: \tau} ; \many{\dot v: \dot \tau} \vdash \parens{ \many v; \dbar v } : \parens{ \many \tau; \dbar \tau } }$
}\\
\multicolumn{2}{c}{
$\Infer{TypeLet}
  {     \Gamma_1; \dot \Gamma_1 \vdash e_1 : \parens{\many \tau; \dbar \tau}
  \quad \Gamma_2, \many{v:\tau}; \dot \Gamma_2, \many{\dot v: \dot \tau} \vdash e_2 : \parens{\many \sigma; \dbar \sigma}}
  {\Gamma_1 \cup \Gamma_2; \dot \Gamma_1 \uplus \dot \Gamma_2 \vdash \letx {\many {v:\tau}; \many{\dot v: \dot \tau}} {e_1} {e_2} : \parens{\many \sigma; \dbar \sigma}}$
}\\
\multicolumn{2}{c}{
$\Infer{TypeUnpack}
  {     \Gamma, \many{w:\tau}; \dot \Gamma \vdash e : \parens{\many \sigma; \dbar \sigma}}
  {\Gamma \cup \{v: \otimes \many \tau\}; \dot \Gamma \vdash \letx {\otimes \many w: \otimes \many \tau;} v e : \parens{\many \sigma; \dbar \sigma}}$
}\\
\multicolumn{2}{c}{
$\Infer{TypeLinUnpack}
  {     \Gamma; \dot \Gamma, \many{\dot w:\dot \tau} \vdash e : \parens{\many \sigma; \dbar \sigma}}
  {\Gamma; \dot \Gamma, \dot v: \otimes \dbar \tau \vdash \letx {;\otimes \dbar w: \otimes \dbar \tau} {\dot v} e : \parens{\many \sigma; \dbar \sigma}}$
}\\
\multicolumn{2}{c}{
$\Infer{TypeApp}
  {     \defx {f} {\many {v:\tau}; \many{\dot v: \dot \tau}} {\many \sigma;\dbar \sigma} e }
  {\many{v:\tau}; \many{\dot v: \dot \tau} \vdash f \parens {\many v; \dbar v} : \parens{\many \sigma; \dbar \sigma}}$
}\\
$\Infer{TypeVar}{~~}
  {v : \tau; \eps \vdash v : \parens{\tau;} }$
&
$\Infer{TypeLit}{~~}
  {\eps; \eps \vdash l : \parens{\R;} }$
\\
$\Infer{TypePrim1}{~~}
  {v:\R; \varepsilon \vdash \sin(v) : \parens{\R;} }$
&
$\Infer{TypePrim2}{~~}
  {\{v_1:\R\} \cup \{v_2:\R\}; \varepsilon \vdash v_1 + v_2 : \parens{\R;} }$
\\
$\Infer{TypeTup}{~~}
  {\cup \{\many{v: \tau}\}; \varepsilon \vdash \otimes \many v : \parens{\otimes \many \tau;} }$
&
$\Infer{TypeLinTup}{~~}
  {\varepsilon; \many{\dot v: \dot \tau} \vdash \otimes \dbar v : \parens{;\otimes \dbar \tau} }$
\\
$\Infer{TypeLinVar}{~~}
  {\eps ; \dot v : \dot \tau \vdash \dot v : \parens{;\dot \tau}}$
&
$\Infer{TypeLinZero}{~~}
  {\eps ; \eps \vdash \dot 0_{\dot \tau} : \parens{;\dot \tau}}$
\\
$\Infer{TypeLinPlus}{~~}
  {\varepsilon; \dot v_1:\dot \tau, \dot v_2:\dot \tau \vdash \dot v_1 + \dot v_2 : \parens{;\dot \tau}}$
& $\Infer{TypeLinMul}{~~}
  {v_1:\R; \dot v_2: \dot \tau \vdash v_1 * \dot v_2 : \parens{;\dot \tau}}$
\\
$\Infer{TypeDup}{~~}
  {\varepsilon; \dot v:\dot \tau \vdash \dup (\dot v) : \parens{;\dot \tau, \dot \tau}}$
&
$\Infer{TypeDrop}
  {     \Gamma; \dot \Gamma \vdash e : \parens{\many \tau;\dbar \tau}}
  {\Gamma; \dot \Gamma \vdash \drop(e) : \parens{;}}$
\\
\multicolumn{2}{c}{
\framebox{$\vdash f$}} \\
\multicolumn{2}{c}{
$\Infer{TypeDef}
  {     \many{v: \tau}; \many{\dot v: \dot \tau} \vdash e : \parens{\many \sigma; \dbar \sigma} }
  { \vdash \defx {f} {\many{v : \tau}; \many{\dot v : \dot \tau}} {\many \sigma; \dbar \sigma} {e} }$
}\\
\end{tabular}
\caption{\Lone\ typing rules.  This is a standard simply-typed system, except that it enforces that all non-linear variables are used at least once, and all linear variables are used exactly once, up to the explicit $\drop$ and $\dup$ operations. The \rl{TypePrim} rules can be extended to a larger set of primitive differentiable maps.}
\label{fig:l1-typerules}
\end{figure}

\subsection{Typing}
\label{sec:l1-typing}

The type system for \Lone\ appears in \Cref{fig:l1-typerules}.
It is a conventional substructural type system that enforces the indexed linearity that \Lone\ is trying to capture.

\subsubsection{Substructural Typing}
The type system is linear on the linear fragment of \Lone, enforcing that every linear variable is used exactly once, except where duplicated or dropped by explicit $\dup$ or $\drop$ operations.
We track linear duplications and deletions because they transpose to addition and zero, respectively.
Our earlier analogy between algebraic and substructural linearity of a polynomial appears here in the rules \rl{TypeLinPlus}, which requires both arguments to be linear, and \rl{TypeLinMul}, which requires exactly the right-hand argument to be linear.

Less conventionally, we also enforce that every non-linear variable is used at least once (by an explicit $\drop$ if needed).
Thus even the non-linear fragment of \Lone\ requires a substructural (if not \emph{linear}, per se) type system.
We track $\drop$ of non-linear variables because we have to charge for it in our cost model (see \Cref{sec:charge-for-drop} for a more complete rationale); we do not track non-linear duplication because we do not charge for $\dup$.

Since our requirement for non-linear variables is to use each at least once, we often pattern-match a non-linear type environment as $\Gamma_1 \cup \Gamma_2$ (e.g., in \rl{TypeLet}) or $\cup \Gamma_i$, implying ``$\Gamma_i$ contains exactly the free variables of $e_i$; the incoming environment $\Gamma$ must be the (not necessarily disjoint) union of the $\Gamma_i$.''
Similarly, when speaking of the linear environment, we write $\dot \Gamma_1 \uplus \dot \Gamma_2$ to imply that in this case the union must be disjoint.
In the end, the actual use of variables is enforced by each leaf rule requiring that the incoming environment contain exactly the variables referenced by that leaf and no others.
In practice, one would implement such a type system by inspecting the free variables of subexpressions rather than by guessing and checking a partition of the environment, but we elide that consideration from the rules for the sake of brevity.

\subsubsection{Polymorphism}

We define linear zero, linear addition, and linear multiplication to be type-indexed---the same operation will type-check with different types in different places as needed (rules \rl{TypeLinZero}, \rl{TypeLinPlus}, \rl{TypeLinMul}).
This is justified because every type that a linear variable can have defines a unique vector space isomorphic to $\R^n$ for some $n$.
Operationally, the zero (respectively, summation and scaling) just happens elementwise, recurring into tuples as needed.

The reason to make linear operations type-indexed like this is that $\drop$ and $\dup$ are naturally type-indexed (i.e., $\drop$ can drop a value of any type), but they transpose to zero and addition, respectively.
We could also type-index non-linear operations, but we don't need to, and it would cease to be justified if \Lone\ were extended to include sum types.
Thus, \rl{TypeLit}, \rl{TypePrim1} and \rl{TypePrim2} call for the $\R$ type.

\subsection{Cost Model}
\label{sec:l1-cost}

The vast automatic differentiation literature often makes efficiency preservation claims, but fails to make the cost model precise.
As it turns out, the cost model required for AD to preserve efficiency is not necessarily obvious.
Here, we state \emph{a} model explicitly, albeit one that has been specifically tailored to make our cost preservation claims hold.
While somewhat unsatisfactory (e.g., we don’t base it on a well-established abstract machine), we think it is still a worthwhile contribution, as it aids the reader in judging the way our method modifies the run-time more precisely.
And, apart from the two finer points we elaborate on below, it could be considered a fairly standard model for an eager language with call-by-value semantics.

Specifically, we choose the following call-by-value cost model for \Lone:
\begin{itemize}
    \item Every non-linear primitive costs 1.
    \item Linear addition and linear multiplication each cost 1 per scalar present in the result.
    \item The $\drop$ operation costs the cost of its subexpression, plus 1 for every scalar of real type (linear or non-linear) present in the argument.
    \item Applying a function costs evaluating its body on the values of its arguments.
    \item The \ttt{let} form costs evaluating its bound expression and then evaluating its body.  This models call-by-value evaluation. 
    \item All other syntactic forms cost 0.
\end{itemize}
Two specific choices bear some elaboration:

\subsubsection{Costly Dropping of Variables}
\label{sec:charge-for-drop}
Why do we charge for $\drop$, when programmers generally think that doing nothing with a variable is free?  Because we need it for work to be preserved.

The fundamental problem with work preservation of reverse-mode AD is that fan-out (i.e., $\dup$) is intuitively costless, but transposes to addition, which is intuitively costly.
We could solve that problem by charging for $\dup$ of linear variables, but then forward differentiation wouldn't be (locally) work-preserving unless we charged for $\dup$ of non-linear variables as well.
Instead, we follow the observation that charging for $\dup$ should have no asymptotic effect, because usually both duplicates would participate in some other costly operation.
In order to make it so, we must charge for $\drop$ (of both linear and non-linear variables), to make it into a costly operation; and doing that indeed suffices for the work preservation proofs to go through.
Another reason why charging for $\drop$ instead of $\dup$ is more natural, is that $\drop$ should be uncommon, or even absent from non-pathological programs since it is only necessary in presence of dead code.

\subsubsection{Costless Tuples}
Why do we not charge for constructing or unpacking tuples, or binding or referencing variables, which would correspond to costly allocation and access of memory in a real implementation?
The main reason is that reverse-mode AD actually does increase the asymptotic memory footprint of a program.
In our system, this is visible in the \rl{ULet} and \rl{UDef} rules in \Cref{sec:partial}, which drastically increase the lifetimes of intermediate non-linear variables (including across function calls).

\subsection{Linearity}
\label{sec:l1-linearity}

Every well-typed \Lone\ expression is \vocab{indexed linear} when viewed as a function from its environment to its return value(s).
We illustrated the dependence and linearity structure in \Cref{fig:dependencies}; now we formalize and prove it in \Cref{thm:linearity}.
Specifically, a \Lone\ expression $e$ defines a collection of algebraically
linear functions---from its free linear variables to its linear outputs---indexed
by values of its free non-linear variables.
The indexing is non-trivial in general because linear multiplication takes a non-linear argument on the left.

\begin{theorem}
\textbf{\Lone\ expressions are indexed linear.}
Consider an expression $e$ of \Lone.
Consider also any set of well-typed values $\many{\elval{x}}$ for the non-linear variables $\many v$ free in $e$, and
any two sets of well-typed values $\dbar{\elval{x}}$ and $\dbar{\elval{y}}$ for the linear variables $\dbar v$ free in $e$; and any scalar $c$.
We write $e[\many{\elval{x}}; \dbar{\elval{x}}]$ for the result of evaluating $e$ with $\many v$ bound to $\many{\elval{x}}$ and $\dbar v$ bound to $\dbar{\elval{x}}$.
We use subscripts as indexing here, so $e[\many{\elval{x}}; \dbar{\elval{x}}]_k$ is the $k$-th result of evaluating $e$ on $\many{\elval{x}}$ and $\dbar{\elval{x}}$.
Then:

\begin{enumerate}
    \item The total work of evaluating $e[\many{\elval{x}}; \many{\delval{x}}]$ is independent of the values $\many{\elval{x}}$ and $\many{\delval{x}}$; \label{claim:work-from-syntax}
    \item The non-linear results are independent of $\mdelval{x}$, i.e., $e[\many{\elval{x}}; \many{\delval{x}}]_k = e[\many{\elval{x}}; \many{\delval{y}}]_k$
    for each non-linear result $k$; \label{claim:non-linear-independence}
    \item The linear results are linear in $\mdelval{x}$, i.e., $e[\many{\elval{x}}; \many{\delval{x} + \delval{y}}]_l = e[\many{\elval{x}}; \many{\delval{x}}]_l + e[\many{\elval{x}}; \many{\delval{y}}]_l$ and $e[\many{\elval{x}}; \many{c \cdot \delval{x}}]_l = c \cdot e[\many{\elval{x}}; \many{\delval{x}}]_l$ for each linear result $l$. \label{claim:linearity}
\end{enumerate}
The equality, addition and scaling here are mathematical operations over reals, and we assume the denotation of arithmetic in \Lone\ to be carried out on infinite-precision reals. Given approximate (e.g., floating-point) semantics, the equalities might not be exact.
\label{thm:linearity}
\end{theorem}

Note that Claim 1 of this theorem, while true for \Lone, is not future-proof.  A variant of the language that supported (non-linear) conditionals or recursion would need to weaken Claim 1 to allow work (including termination) to depend on the non-linear values $\many{\elval{x}}$, though work should remain independent of the linear values $\many{\delval{x}}$.

\begin{proof}
By structural induction, first on the program and then on the syntax of expressions and function bodies.
Work-independence (Claim 1) follows immediately from the lack of conditional control flow in the language.

The proofs of Claims 2 and 3 are similar enough that we discuss them jointly.
The interesting case is multi-value \ttt{let}, so we consider it first. 
The inductive hypothesis asserts that the non-linear values returned by the bound expression are independent of the linear free variables thereof, and the returned linear values satisfy the additive and scaling laws.
Since the set of free variables of the body expression is a union of free variables of the entire expression and the let bound values, we can again invoke the inductive hypothesis to arrive at the desired conclusion.

Tuple unpacking is entirely analogous to the multi-value \ttt{let}.
Function application follows directly from induction on the function body.\footnote{This is why we need the outer induction on the whole program rather than just structural induction on the syntax of a single expression.}
The $\drop$ form is trivial, and all other forms can be verified directly because they do not have subexpressions.
In particular, indexed linearity depends upon the only primitive operations available to the linear fragment of \Lone\ being zero, plus, and scaling by a non-linear constant; and the presence of actually non-linear primitives in the non-linear fragment of \Lone\ is what prevents us from strengthening the linearity claims to the non-linear variables.
\end{proof}

\begin{corollary}
\label{thm:zero-values}
Every closed term in \Lone\ returns 0 for all its linear results.
\end{corollary}

\begin{proof}
If the term $e$ is closed, we can vacuously scale all of the free linear variables of $e$ without changing the evaluation of $e$.  All the linear results must therefore be 0.  Formally, Claim 3 gives us $e[;]_l = 2 e[;]_l = 0$ for every linear result position $l$ of $e$.
\end{proof}

In particular, for the statement of \Cref{thm:linearity} not to be vacuous, one must allow $\mdelval{x}$ and $\mdelval{y}$ to be non-zero without asking where those non-zero values came from.
This is ok, though, because we are generally interested in \emph{open} terms of \Lone, such as the bodies of \Lone\ functions.  In that case, we can assume the non-zero linear values are provided externally through the substitution of free variables.

\begin{lemma}
The work a \Lone\ function or expression does is at least the number of real scalars in its linear inputs, less the number of real scalars in its linear outputs.
\label{lem:dead-code}
\end{lemma}

\begin{proof}
The only linear operations in \Lone\ that return fewer scalar results than they consume are \ttt{+} and $\drop$, each of which costs 1 per net scalar consumed.  Since all linear variables must be either consumed or returned, the result follows.
\end{proof}

We state this fact here because it's a curious general property of our cost model.  We will use it only to argue that eliminating useless code (e.g., by short-circuiting in the rule for transposing $\drop(e)$) does not increase work (!).

\subsection{Linearity Erasure}
\label{sec:linearity-erasure}

Because we will occasionally need it, we define a linearity erasure transform $\Le$ on \Lone.
Linearity erasure is just what it sounds like: remove the linearity information by turning all linear variables and expressions into non-linear ones.
Syntactically, this just means
\begin{itemize}
    \item For every type at which linear addition, multiplication, and zero are used in the program, define a fresh function that carries that operation out in the non-linear fragment on \Lone\ (using unpacking and repacking as needed);
    \item Replace every linear addition, multiplication, and zero with a call to the corresponding function;
    \item Rename every reference to a variable produced by $\dup$ to refer to the input of that $\dup$ instead;
    \item Remove all $\dup$ operations; and
    \item Move all the semicolons all the way to the right.
\end{itemize}
In the interest of space, we do not give rules for this transform.

\begin{lemma}
Linearity erasure preserves the semantics of an expression, and preserves work in our cost model.
\end{lemma}
\begin{proof} By inspection. \end{proof}

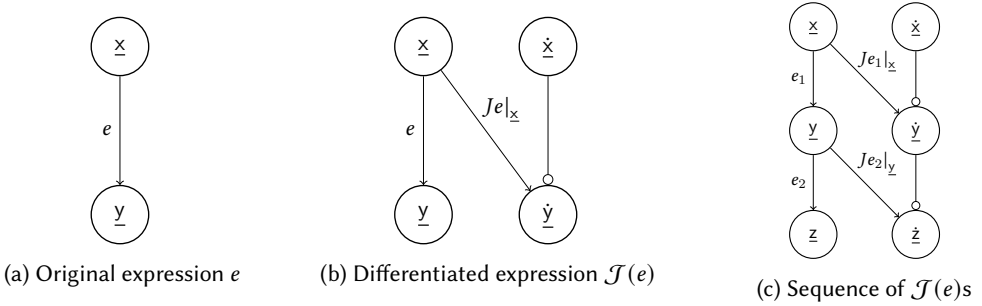
\begin{figure}
\centering
\begin{minipage}{.28\columnwidth}%
\centering
\scalebox{.85}{%
\begin{tikzpicture}[
  y=-1cm,
  scale=1.3,
  decoration=snake,
  xx/.style={circle, draw=black, inner sep=0pt, minimum size=9mm, semithick},
  dd/.style={}
]
\node[xx] (xp) at (0, 0) {$\melval{x}$};
\node[xx] (yp) at (0, 2) {$\melval{y}$};

\draw [->] (xp) -- (yp) node[midway,left] {$e$};
\end{tikzpicture}
}%
\subcaption{Original expression $e$}
\label{fig:dataflow:expr}
\end{minipage}
\hfill
\begin{minipage}{0.34\columnwidth}
\centering
\scalebox{.85}{%
\begin{tikzpicture}[
  y=-1cm,
  scale=1.3,
  decoration=snake,
  xx/.style={circle, draw=black, inner sep=0pt, minimum size=9mm, semithick},
  dd/.style={}
]
\node[xx] (xp) at (0,   0) {$\melval{x}$};
\node[xx] (xt) at (1.5, 0) {$\mdelval{x}$};
\node[xx] (yp) at (0,   2) {$\melval{y}$};
\node[xx] (yt) at (1.5, 2) {$\mdelval{y}$};

\draw [->] (xp) -- (yp) node[midway,left] {$e$};
\draw [->] (xp) -- (yt) node[midway,above] { $\ \ \ \ \ \left.J e\right|_{\melval{x}}$ };
\draw [-o] (xt) to (yt);
\end{tikzpicture}
}%
\subcaption{Differentiated expression $\J(e)$}
\label{fig:dataflow:jvp}
\end{minipage}
\hfill
\begin{minipage}{.30\columnwidth}
\centering
\scalebox{.7}{%
\begin{tikzpicture}[
  y=-1cm,
  scale=1.3,
  decoration=snake,
  xx/.style={circle, draw=black, inner sep=0pt, minimum size=9mm, semithick},
  dd/.style={}
]
\node[xx] (xp) at (0,   0) {$\melval{x}$};
\node[xx] (xt) at (1.5, 0) {$\mdelval{x}$};
\node[xx] (yp) at (0,   1.5) {$\melval{y}$};
\node[xx] (yt) at (1.5, 1.5) {$\mdelval{y}$};
\node[xx] (zp) at (0,   3) {$\melval{z}$};
\node[xx] (zt) at (1.5, 3) {$\mdelval{z}$};

\draw [->] (xp) -- (yp) node[midway,left] {$e_1$};
\draw [->] (xp) -- (yt) node[midway,above] { $\ \ \ \ \ \left.J e_1\right|_{\melval{x}}$ };
\draw [-o] (xt) to (yt);
\draw [->] (yp) -- (zp) node[midway,left] {$e_2$};
\draw [->] (yp) -- (zt) node[midway,above] { $\ \ \ \ \ \left.J e_2\right|_{\melval{y}}$ };
\draw [-o] (yt) to (zt);
\end{tikzpicture}
}%
\subcaption{Sequence of $\J(e)$s}
\label{fig:dataflow:jvpstack}
\end{minipage}
\caption{Dataflow diagrams of
(a) a purely non-linear expression in \Lone{} and
(b) its derivative,
plus (c) a derivative made of a sequential composition
stacked vertically.
The primal (non-linear) results
of a derivative depend only only on its primal (non-linear) inputs,
whereas the tangent results depend
non-linearly on the primal (non-linear) inputs
and linearly on the tangent (linear) inputs.
Note the distinction between $\J$, the code transform we define, and $J$, the Jacobian of an expression $e$.}
\label{fig:dataflow}
\end{figure}

\section{Automatic differentiation}
\label{sec:jvp}

Our first step is a code transformation $\J$ performing forward-mode automatic differentiation\footnote{$\J$ corresponds to the \texttt{jax.jvp} transform from JAX.}.
Specifically, $\J$ transforms the non-linear fragment of \Lone\ into the full language as follows.
For any purely non-linear $e$ in \Lone, and input $\many{\elval{x}}$, $\J(e)[\many{\elval{x}}; \many{\delval{x}}]$ has non-linear results computing the same output as $e[\many{\elval{x}};]$, and linear results computing the directional derivative of $e$ at the point $\many{\elval{x}}$ in the direction $\many{\delval{x}}$.
We illustrate the data flow produced by $\J$ in \Cref{fig:dataflow}.
Here we view the expression $e$ as a function from its free variables $\many v$ to its results.
A code example appears in \Cref{fig:jvp-sin}.

To compute a forward derivative end-to-end, we need to supply the initial direction of differentiation (1 for an $\R \to \R^m$ primal function).
Since we cannot write a non-zero linear literal in well-typed \Lone\ program, we have to use the linearity erasure transform $\Le$ from \Cref{sec:linearity-erasure}.  For $F: \R \to \R^m$,
\[ F(x), \frac{d}{dx} F(x) = \Le(\J(F))(x, 1; ). \]

Part of the point of this whole exercise is that $\J$ will be very familiar to students of automatic differentiation;
it's just a forward-mode transformation that computes tangents in tandem with primals.
The rules in \Cref{fig:jvp} are completely standard, with the one wrinkle that the derivative is emitted into the linear fragment of \Lone.  Ergo, the result being well-typed (substructurally) in \Lone\ means that we have a proof that the derivative of a function is algebraically linear (with respect to the direction in which the derivative is taken).

\begin{figure}\codeexample
\begin{minipage}{0.45\textwidth}
\begin{align*}
\defxArrow
  {& \ttt{f}^{1,0}_{1,0}}
  {u: \R; }
  {\R; } {}
  {\\
  & \letx {v: \R; } { \sin(u) } {\\
  & \letx {w: \R; } { -v } {\\
  & \parens{w; }}}}
\end{align*}
\end{minipage}
$\transformsto{\J}$
\begin{minipage}{0.45\textwidth}
\begin{align*}
\defxArrow
  {& {\ttt{f}^\J}^{1,1}_{1,1}}
  {u: \R; \dot u: \R}
  {\R; \R} {}
  {\\
  & \letx {v: \R; } { \sin(u) } {\\
  & \letx {du: \R; } { \cos(u) } {\\
  & \letx {; \dot v: \R} { du * \dot u } {\\
  & \letx {w: \R; } { -v } {\\
  & \letx {; \dot w: \R} { -\dot v } {\\
  & \parens{w; \dot w}}}}}}}
\end{align*}
\end{minipage}
\caption{Example of forward differentiation with $\J$, transforming $f(u) = -\sin(u)$ on the left into $\J(f)$, which computes both the original $f$ and its directional derivative.  The input function $f$ must be coded in \Lone\ as all non-linear.  The result uses both the linear and non-linear fragments of \Lone.}
\label{fig:jvp-sin}
\end{figure}

This is the first step in our separation of concerns: differentiation proper is confined to the relatively simple $\J$, whereas arranging to run the derivative backward to obtain reverse-mode AD is the province of $\Pa$ and $\T$ (in \Cref{sec:partial} and \Cref{sec:transposition}, respectively).
The latter two know only about linearity, not about differentiation, nor about any special relationship between the non-linear (primal) and the linear (tangent) computations.
The type system of \Lone\ is the abstraction boundary that lets differentiation and transposition be implemented independently.

\subsection{Details}

\subsubsection{Non-linear Input}
We define $\J$ to operate only on the purely non-linear subset of \Lone.  This avoids perturbation confusion problems: all linear variables in \Lone\ refer to the same perturbation.
We can still get nested AD by alternating $\J$ and the linearity erasure transform $\Le$.

\begin{figure}\footnotesize
\renewcommand{\arraystretch}{3}
\begin{tabular}{cc}
\multicolumn{2}{c}{
\framebox{$\jvp \Sigma e \defunc e'$}} \\
\multicolumn{2}{c}{
$\Infer{JRet}{\many v\ \trm{distinct}}
  {\jvp {\many{v \to \dot v}} {\parens{\many v;}} \defunc \parens{\many v; \dbar v}}$}\\
\multicolumn{2}{c}{
$\Infer{JLet}
  {     \jvp {\Sigma_1, \many{u \to \dot w}} {e_1} \defunc e_1'
  \quad \jvp {\Sigma_2, \many{u \to \dot z}, \many{v \to \dot v}} {e_2} \defunc e_2'\\ \\
  \quad \dbar{v}, \dbar{w}, \dbar{z}\ \trm{fresh}
  \quad \many{u:\tau}\ \trm{free in both $e_1$ and $e_2$}
  }
  {\jvp {\Sigma_1 \uplus \Sigma_2 \uplus \{\many{u \to \dot u}\}} {\letx {\many{v:\sigma};} {e_1} {e_2} } \\ \defunc
  \underbar{$\letx {;\dot w:\typej{\tau}, \dot z:\typej{\tau}} {\parens{;\dup(\dot u)}\ \mbox{} } $} {
    \letx {\many{v:\sigma}; \many{\dot v:\typej{\sigma}}} {e_1'} {e_2'} } }$}\\
\multicolumn{2}{c}{
$\Infer{JUnpackA}
  {     \jvp {\Sigma, \many{w \to \dot w}} e \defunc e'
  \quad \dbar w\ \trm{fresh}
  \quad v\ \trm{not free in $e$}
  }
  {\jvp {\Sigma, v \to \dot v} {\letx {\otimes \many w: \otimes \many \tau} v e} \\ \\ \defunc
  \letx {\otimes \many w: \otimes \many \tau} v {\letx {\otimes \dbar w: \otimes \many{\typej{\tau}}} {\dot v} e'}}$}\\
$\Infer{JApp}{\many v\ \trm{distinct}}
  {\jvp {\many{v \to \dot v}} {f\parens{\many{v};}} \defunc f^\J \parens{\many v; \dbar v}}$ &
$\Infer{JVar}{~~}
  {\jvp {v \to \dot v} v \defunc \parens{v; \dot v}}$\\
$\Infer{JTup}{\many v\ \trm{distinct}}
  {\jvp {\many{v \to \dot v}} {\otimes \many v} \defunc \parens{\otimes \many v; \otimes \dbar v}}$ &
$\Infer{JLit}{~~}
  {\jvp {\varepsilon} {l} \defunc \parens{l; 0_{\typej{\R}}}}$\\
\multicolumn{2}{c}{
$\Infer{JPrimSin}{~~}
  {\jvp {v \to \dot v} {\sin(v)} \defunc \parens{\sin(v); \cos(v) * \dot v} }$}\\
\multicolumn{2}{c}{
$\Infer{JPrimMul}{v_1, v_2\ \trm{distinct}}
  {\jvp {v_1 \to \dot v_1, v_2 \to \dot v_2} {v_1 * v_2} \defunc
   \parens{v_1 * v_2; v_1 * \dot v_2 + v_2 * \dot v_1} }$}\\
\multicolumn{2}{c}{
$\Infer{JDrop}
  {     \jvp \Sigma {e} \defunc e'}
  {\jvp {\Sigma} {\drop (e)} \defunc \drop (e')}$}\\
\multicolumn{2}{c}{
\framebox{$\J(f) \defunc f^\J$}} \\
\multicolumn{2}{c}{
$\Infer{JDef}
  {     \jvp {\many{v \to \dot v}}{e} \defunc e'
  \quad \dbar v\ \trm{fresh}}
  {\J (\defx f {\many{v:\tau};} {\many \sigma;} e ) \defunc
    \defx {f^\J} {\many{v:\tau}; \many{\dot v: \typej{\tau}}} {\many{\sigma};\many{\typej{\sigma}}} {e'}}$}
\\
\end{tabular}
\caption{Forward-mode automatic differentiation in \Lone.
The constructed expression $e'$ emits primal and tangent results together as a multi-value return.
The argument $\Sigma$ maps primal variable names to the corresponding tangent variable names.
Note in the result of rule \rl{JPrimMul}, the multiplication on the left is non-linear, whereas the two multiplications on the right are right-linear.
Several near-redundant rules are omitted; see text.
}
\label{fig:jvp}
\end{figure}

\subsubsection{Introducing $\dup$}
The non-linear fragment of \Lone\ allows variables to be referenced multiple times, but the linear fragment does not.
Ergo, whenever a non-linear variable is re-used, we must introduce $\dup$ operations on the corresponding linear variables.
The \rl{JLet} rule in \Cref{fig:jvp} spells this out: the rule detects which variables $\many u$ are repeated; constructs the needed fresh variables $\many {\dot w, \dot z}$; introduces the duplication operations $\many{\dot w, \dot z = \dup(\dot u)}$; and takes care that the tangent of $\many{u}$ in $e_1$ is $\dbar w$, whereas in $e_2$ it's $\dbar z$.

One instance seeming sufficient, \Cref{fig:jvp} omits similar manipulations for the other rules.  The full system has a rule \rl{JUnpackB} for the case where the tuple variable $v$ does appear free in the body $e$; it differs from \rl{JUnpackA} only in that the binding for $v$ is propagated to the recursive call, and the tangent $\dot v$ is duplicated for its appearance in the transformed body.
Likewise, \rl{JRet}, \rl{JApp}, \rl{JTup}, and \rl{JPrimMul} have variants that add the necessary $\dup$ calls if any of the non-linear $\many v$ repeat.

\subsubsection{Other Omissions}
We also omit rules for the other non-linear primitives.  They differ from \rl{JPrimSin} and \rl{JPrimMul} only in the derivatives they must insert into the result; we trust that a unary and a binary example are adequately clear.
Finally, there are no rules for linear syntax because $\J$ operates on the non-linear fragment of \Lone.

\subsubsection{Extra Non-linear Computations}
Note that $\J$ in general introduces non-linear computations that are not needed for the original non-linear program---these are the actual partial derivatives.  For example, in \rl{JPrimSin} we see a new $\cos(v)$.  This is bound to a non-linear variable, but used only by the linear result of $\J(\sin)$.
\Cref{fig:ad-full} elided these extra non-linear computations, hiding them in the red arrows between the blue primal blocks and orange derivative blocks.  Our downstream unzipping transformation $\Pa$ treats them as part of the blue blocks, since they are non-linear.

\subsubsection{Tangent Types}
The rules in \Cref{fig:jvp} depend on assuming that any \Lone\ type $\tau$ is also a suitable type for its tangent.
As we have defined \Lone\ to have only real scalars and nested tuples thereof, this proposition is costless; but the correct choice of tangent type becomes more interesting in a language with sum types.
One conventional choice is to still use $\tau$ itself as the type of tangents to $\tau$, though that loses the information that both the primal and the tangent value are necessarily the same variant.
Retaining that information requires a dependent type system; we look forward to the community working out a satisfying one for this purpose.

\subsection{Formalization}

Our theorem about forward differentiation is conventional.
The interesting thing about it is that the $\J$ transformation puts the derivative into the linear fragment of \Lone.
We therefore have a proof obligation that automated derivatives are substructurally linear; but otherwise this is a repetition of known results that forward differentiation is correct, and work-preserving up to a constant factor.

\begin{theorem}
\textbf{Automatic differentiation is total, correct, and work-preserving.}
Consider an expression $e$, well-typed in $\Gamma; \{\}$ in the non-linear fragment of \Lone.
Let $\melval{x}$ be well-typed values for the (necessarily non-linear) free variables $\many{v: \tau}$ of $e$.
Allocate fresh names $\dbar v$ of types $\many{\typej{\tau}}$ for the tangents, and
let $\mdelval{x}$ be well-typed values for them.  Then:
\begin{enumerate}
    \item $e' = \J(\many{v \to \dot v} | e)$ exists and is well-typed in $\Gamma; \many{\dot v : \typej{\tau}}$ in \Lone;
    \item $e'$ does no more than a constant times more work than $e$.  The constant does not depend on $e$, but does depend on the set of primitives and their derivative rules;
    \item The non-linear outputs of $e'[\melval{x}; \cdot]$ are equal to the outputs of $e[\melval{x};]$; and
    \item The linear outputs of $e'[\melval{x}; \mdelval{x}]$ are equal to the matrix-vector product of the Jacobian of $e$ at $\melval{x}$ with the vector $\mdelval{x}$, provided this is true of the primitive derivatives. %
\end{enumerate}
To define the Jacobian for Claim 4, we view the expression $e$ as a function from its free variables $\many v$ to its results.
\label{thm:jvp}
\end{theorem}

We assume these claims hold for all the primitives in a given implementation of \Lone.  In particular, the constant in Claim 2 may not actually be well-determined in a large and user-extensible AD system, because it depends on the cost behavior of every primitive derivative on every set of arguments.  We content ourselves in this paper with proving that $\J$ composes those primitive derivatives correctly, in the sense of introducing no correctness or work errors that are not present in the primitive set already. 

\begin{proof}
The proof proceeds by structural induction on the program.  We work out the \rl{JLet} rule as a template for the others.  The first subtlety is that there may be some non-linear variables, $\many u$ in the notation, that are free in both the bound expression $e_1$ and the body $e_2$.  To construct the parallel linear expression, we must insert explicit $\dup$s for their corresponding linear variables $\dbar u$.  This is notated as $\many{\letx {;\dot w, \dot z} {\parens{; \dup(\dot u)}} {\mbox{}}}$ in \rl{JLet}, implying one \ttt{let} and one $\dup$ per duplicated variable among the $\dbar u$.  This duplication is also why the map from primal variables to their corresponding tangent variables may be non-trivial and needs to be maintained explicitly.  Work preservation (Claim 2) follows because the $\dup$s introduced by \rl{JLet} cost 0 in our cost model.

\newcommand{\jacat}[2]{\left( \left. J{{#1}} \right|_{{#2}} \right)}
Non-linear correctness (Claim 3) follows from the inductive hypothesis because we reconstruct the non-linear part of the computation exactly in the result $\J(e)$.  Linear correctness (Claim 4) is a calculation:
\begin{align*}
    \J(e)[\melval{x}; \mdelval{x}]
    & = \J(e_2)[\melval{x}_{2}, e_1[\melval{x};]; \mdelval{x}_{2}, \J(e_1)[\melval{x}_{1}; \mdelval{x}_{1}] ]
      & \trm{definition of $\J(e)$} \\
    & = \jacat{e_2}{\melval{x}_{2}, e_1[\melval{x};]}(\mdelval{x}_{2}, \J(e_1)[\melval{x}_{1}; \mdelval{x}_{1}]) 
      & \trm{induction on $e_2$} \\
    & = \jacat{e_2}{\melval{x}_{2}, e_1[\melval{x};] }\left( \mdelval{x}_{2}, \jacat{e_1}{\melval{x}_{1}}(\mdelval{x}_{1}) \right)
      & \trm{induction on $e_1$} \\
    & = \jacat{e}{\melval{x}}(\mdelval{x}).
      & \trm{chain rule on $e_2 \circ e_1$}
\end{align*}
The chain rule is for the composition of multivariate functions $e_2 \circ e_1$ augmented with side inputs $\melval{x}_{2}$ and their tangents $\mdelval{x}_{2}$.  Here we somewhat abuse our own notation, writing $\jacat{e}{\cdot}$ for the Jacobian of $e$ at a point, $\J(e)[\cdot; \cdot]$ for the \emph{linear} results of evaluating $\J(e)$, and $e[\cdot;]$ with no linear inputs for the \emph{non-linear} results of evaluating $e$.  The reader must also understand the index $1$ to refer to variables free in $e_1$ and $2$ to refer to those free in $e_2$, with any overlap among linear ones taken care of by the above-mentioned chain of \ttt{let} and $\dup$.

Induction on the program (including previously processed definitions) comes in when checking \rl{JApp}, because we rely on correctness of the transformed callee for correctness of the transformed call site.
The handling of \rl{JUnpack} is analogous to \rl{JLet}, and the other rules are straightforward.
\end{proof}

\section{Unzipping}
\label{sec:partial}

We cannot directly transpose arbitrary \Lone\ programs.  Why not?  In the transposed result we are trying to run the non-linear part forward and the linear part backward.
However, the first sub-expression of a general \Lone\ program may produce both linear and non-linear results.
Do we run that subexpression before or after the remainder?

Fortunately, the non-linear part of a \Lone\ expression
cannot depend on any of its linear inputs.  Ergo, it should be possible to hoist it above the linear part (this section) to just transpose the latter (\Cref{sec:transposition}).

We begin by defining what we want:
\begin{definition}[\Lonestrict]
\label{def:linear-b}
\Lonestrict\ is a sublanguage of \Lone\ with two restrictions:
\begin{itemize}
    \item All functions and expressions of \Lonestrict\ must return either all-linear or all-nonlinear results---no mixing.
    \item Expressions of \Lonestrict\ that produce non-linear results must not read linear variables, not even to $\drop$ them.
\end{itemize}
\end{definition}

\begin{figure}\footnotesize
\renewcommand{\arraystretch}{3}
\begin{tabular}{c}
\framebox{$\opP{\envNL; \envTan}{e} \defunc E; e'; \dot e'$} \\
$\Infer{URet}{~~}
  {\opP {\envNL;\envTan} {\parens{\many v; \dbar v}}
  \defunc \hole ; \parens {\many v;}; \parens {;\dbar v}}$\\
$\Infer{ULet}
  {     \opP {\envNL;\envTan} {e_1} \defunc E_1; e'_1; \dot e'_1
  \quad \opP {\envNL, v_i: \tau_i; \envTan, \dot v_j: \dot \tau_j} {e_2} \defunc E_2; e'_2; \dot e'_2}
  {\opP {\envNL;\envTan} {\letx{\many{v: \tau}; \many{\dot v: \dot \tau}}{e_1}{e_2}}
  \defunc E_1, \letx{\many{v:\tau};}{e_1'}{E_2};\ e_2';\ 
  \letx {; \many{\dot v: \dot \tau}}{\dot e_1'}{\dot e_2'}}$\\
$\Infer{UUnpack}
  {     \opP {\envNL, \many{w: \tau}; \envTan} e \defunc E; e'; \dot e'}
  {\opP {\envNL;\envTan} {\letx{\otimes \many{w}: \otimes \many{\tau};} v e}
  \defunc \letx {\otimes \many{w}: \otimes \many \tau} {v} {E};\ e';\ {\dot e'}}$\\
$\Infer{ULinUnpack}
  {     \opP {\envNL, \many{w: \tau}; \envTan} e \defunc E; e'; \dot e'}
  {\opP {\envNL;\envTan} {\letx{; \otimes \dbar w: \otimes \dbar \tau} {\dot v} e}
  \defunc E;\ e';\ \letx {\otimes \dbar w : \otimes \dbar \tau} {\dot v} {\dot e'}}$\\
$\Infer{UApp}
  {     \vdash \nonlin f \parens{\cdot;} \to \parens{\many \tau, \sigma;}
  \quad \many u, w\ \trm{fresh}}
  {\opP {\envNL;\envTan} {f \parens{\many v; \dbar v} }
  \defunc \letx {\many {u: \tau}, w: \sigma;} {\nonlin{f}(\many v;)} {\hole};\ \parens{\many u;};\ \lin{f}(w; \dbar v)}$\\
$\Infer{ULeaf}
  {e\ \trm{non-linear with no sub-expressions}}
  {\opP {\envNL;\envTan} e
  \defunc \square; e; \parens{;}}$\\
$\Infer{ULinLeaf}
  {\dot e\ \trm{linear with no sub-expressions}}
  {\opP {\envNL;\envTan} {\dot e}
  \defunc \square; \parens{;}; \dot e}$\\
$\Infer{UDrop}
  {\Pa(e) \defunc E; e'; \dot e'}
  {\Pa(\drop (e)) \defunc E; \drop (e'); \drop (\dot e') }$\\
\framebox{$\Pa(f) \defunc \nonlin{f}, \lin{f}$} \\
$\Infer{UDef}{\opP {v_i:\tau_i; \dot v_j:\dot \tau_j} e \defunc E; e'; \dot e'
  \quad \many{w:\pi}\ \trm{are all the names bound by}\ \many v\ \trm{and}\ E\ \trm{and read by}\ \dot e'
  \quad u\ \trm{fresh}}
  {\Pa (\defx {f} {\many{v:\tau}; \many{\dot v:\dot \tau}} {\many \sigma; \dbar \sigma} e) \defunc
  \defx {\nonlin{f}} {\many{v:\tau};} {\many \sigma, \otimes \many \pi;} {E \parens{e', \otimes \parens{\many w}}}, \\
  \defx {\lin{f}} {u : \otimes \many{\pi}; \many{\dot v:\dot \tau}} {;\dbar \sigma}
    { \letx {\otimes \many w: \otimes \many \pi} {u} {\dot e'} }
  }$\\
\end{tabular}
\caption{Unzipping converting \Lone\ to \Lonestrict.
Unzipping transforms a mixed linear and non-linear expression $e$ into a shared non-linear context $E$, a non-linear result expression $e'$, and a linear result expression $\dot e'$.
Note that we always unzip recursively, because \Lonestrict\ does not allow non-linear expressions to read linear variables, even to $\drop$ them.
Most language primitives are handled by the generic rules \rl{ULeaf} and \rl{ULinLeaf}.
}
\label{fig:partial-evaluation}
\end{figure}

\Cref{fig:partial-evaluation} gives the transformation $\Pa$\footnote{$\Pa$ is not surfaced as a public API in JAX, but it is used internally by the AD implementation. The relevant implementation can be found in the file \texttt{jax/interpreters/partial\_eval.py}.} that ``unzips'' any \Lone\ expression or function into a pair of \Lonestrict\ expressions or functions---one that computes only linear results and one that computes only non-linear results.
\Cref{fig:unzip-sin} shows the unzip of an example function.

The result of $\Pa$ on an expression is a bit subtle: we wish for $\Pa$ to split the input expression $e$ into a non-linear fragment $e'$ and a linear fragment $\dot e'$.  However, $e$ may compute a (non-linear) intermediate value that will be used for both linear and non-linear results.  To allow $e'$ and $\dot e'$ to both use that value without having to recompute it, $\Pa(e)$ also produces a context $E$ of let bindings, which represent all the intermediate (non-linear) computations that $e'$ and $\dot e'$ may share.  All the actual work of the non-linear fragment ends up in the context $E$.  The rule \rl{UDef} exposes the variables bound by $E$ as outputs of the emitted non-linear fragment $\nonlin f$, so they can be read without recomputation by the linear fragment $\lin f$.

\begin{figure}\codeexample
\begin{minipage}{0.45\textwidth}
\begin{align*}
\defxArrow
  {& {\ttt{f}^\J}^{1,1}_{1,1}}
  {u: \R; \dot u: \R}
  {\R; \R} {}
  {\\
  & \letx {v: \R; } { \sin(u) } {\\
  & \letx {du: \R; } { \cos(u) } {\\
  & \letx {; \dot v: \R} { du * \dot u } {\\
  & \letx {w: \R; } { -v } {\\
  & \letx {; \dot w: \R} { -\dot v } {\\
  & \parens{w; \dot w}}}}}}}
\end{align*}
\end{minipage}
$\transformsto{\Pa}$
\begin{minipage}{0.45\textwidth}
\begin{align*}
\defxArrow
  {& \nonlin{\ttt{f}^\J}^{1,0}_{2,0}}
  {u: \R; }
  {\R, \otimes \R; } {}
  {\\
  & \letx {v: \R; } { \sin(u) } {\\
  & \letx {du: \R; } { \cos(u) } {\\
  & \letx {w: \R; } { -v } {\\
  & \letx {t: \otimes \R; } {\otimes \parens{du}} {\\
  & \parens{w, t;}}}}}}\\
\defxArrow
  {& \lin{\ttt{f}^\J}^{1,1}_{0,1}}
  {t: \otimes \R; \dot u: \R}
  {; \R} {}
  {\\
  & \letx {\otimes du; } {t} {\\
  & \letx {; \dot v: \R} { du * \dot u } {\\
  & \letx {; \dot w: \R} { -\dot v } {\\
  & \parens{; \dot w}}}}}
\end{align*}
\end{minipage}
\caption{Example of unzipping.  On the left we have the derivative $\J(f)$ of $f(u) = -\sin(u)$ from \Cref{fig:jvp-sin}, and on the right the non-linear and linear functions it unzips to, $\nonlin{\J(f)}$ and $\lin{\J(f)}$.
In AD terminology, $\nonlin{\J(f)}$ is the forward phase, which computes the original function $f$ and returns all the intermediates, or the \vocab{tape}, (in this case the variable $du$) needed for the derivative.
The function $\lin{\J(f)}$ is called the linearization in \citet{frostig2021decomposing}.
It consumes the tape produced by the forward phase and computes (in substructurally and therefore algebraically linear fashion) the forward derivative of $f$; in \Cref{fig:transpose-sin} (\Cref{sec:transposition}) we will transpose it to obtain the reverse phase for computing the derivative of $f$ in reverse mode.
}
\label{fig:unzip-sin}
\end{figure}

The context $E$ is the ``tape'' (as it is often called \citep{griewank2008evaluating}) that gives reverse-mode AD its work-preservation.
In the example of \Cref{fig:unzip-sin} it consists of the single variable $dx$.
One of our contributions is disentangling the construction of the tape from AD proper.
The time-space tradeoff the tape represents is handled by the unzipping transformation, and neither (forward) differentiation nor transposition need concern themselves with it.

\subsection{Relation to Partial Evaluation}

Unzipping can be seen as a form of partial evaluation (PE).  Specifically, $\Pa$ symbolically partially evaluates $f$ with respect to its non-linear arguments, leaving $\lin f$ as the residual.  The first difference with traditional PE \citep{jones1993partialeval} is that we do not require concrete values for the non-linear arguments of $f$.  This is because our goal for unzipping is not to use static information to make the residual faster, but simply to separate the linear fragment so it can be transposed.

The second difference is that $\Pa$ always produces exactly two specializations of any function definition, whereas general partial evaluation may produce arbitrarily many.
This is because we are always unzipping exactly the linear from the non-linear variables of every function, and \Lone\ does not permit the same function definition to be used with different arguments being treated linearly or non-linearly (even if the underlying computation is mathematically linear in different sets of its inputs).

We hope this identification of unzipping with partial evaluation may spur further advances in both PE and AD.  In particular, checkpointing strategies developed for memory efficiency by the AD community might prove useful for general-purpose partial evaluation in memory-limited settings.
In our system, checkpointing materializes thus: our \rl{ULet} rule remains correct (though no longer work-preserving) if modified not to save a non-linear intermediate, and recompute it in the linear expression instead.
Our transposition transformation $\T$ (\Cref{sec:transposition}) will then retain the non-linear recomputation while reversing the direction of the linear computation, implementing checkpointing.
Perhaps similar ideas can be applied to partial evaluation more broadly.

\subsection{Formalization}

Given an expression $e$ well-typed in $\Gamma; \dot \Gamma$ in \Lone, our transformation $\Pa(e)$ produces a context $E$ (syntax in \Cref{fig:linear-a}), a non-linear expression $e'$, and a linear \vocab{residual} expression $\dot e'$.  We need the free non-linear variables $\many w$ of the expression $\dot e'$, with their types $\many \pi$.  Then we define
\begin{itemize}
    \item The \vocab{tape} $\many w$.  These are variables carrying the needed non-linear information from $\Gamma$ and $E$ to $\dot e'$.
    \item The \vocab{non-linear partial} $E\parens{e', \many w}$.  This is the expression formed by allocating fresh variables $\many u$ for the results of $e'$ and plugging $\letx {\many u;} {e'} {\parens{\many u, \many w;}}$ into the hole in $E$.  The non-linear partial is a purely non-linear expression that returns the results of $e'$ together with the tape.
    \item The \vocab{reconstruction} $E\parens{e', \many w; \dot e'}$.  Given fresh variables $\many u$ and $\dbar u$ for all the results of $e$, this is the expression $\letx {\many u, \many w;} {E \parens{e', \many w}} {\letx {; \dbar u} {\dot e'} {\parens{\many u; \dbar u}}}$ that should be equivalent to $e$ if the non-linear partial and the residual are correct.  Note that here we bind $\many w$ to themselves: the point is that we arrange for the context $E$ not to be in scope in $\dot e'$, and carry all the needed information explicitly in the tape $\many w$. 
\end{itemize}

\begin{theorem}
\textbf{Unzipping is total, correct, and work-preserving.}
Consider any \Lone\ expression $e$ well-typed in $\Gamma; \dot \Gamma$ with $\many w$ and $\many \pi$ as above.  Then

\begin{enumerate}
    \item $E; e'; \dot e' = \Pa(e)$ exist;
    \item The non-linear partial $E \parens{e', \many w}$ is well-typed in $\Gamma; \{\}$ in \Lonestrict;
    \item The residual $\dot e'$ is well-typed in $\many {w: \pi}; \dot \Gamma$ in \Lonestrict; and
    \item The reconstruction $E\parens{e', \many w; \dot e'}$ is well-typed in $\Gamma; \dot \Gamma$ in \Lone\ and computes the same result with the same work as $e$.
\end{enumerate}
\label{thm:unzip}
\end{theorem}

\begin{proof}
The proof consists of our usual structural induction.
One subtlety is well-typedness of the non-linear partial $E \parens{e', \many w}$ when $e$ is a mixed $\ttt{let}$ expression (rule \rl{ULet}).  We can hoist $E_2$ and $e_2'$ past the linear variables $\dbar v$ because recursive unzipping will ensure that $E_2$ and $e_2'$ do not read any linear variables.  We also need to check that the non-linear partial uses every non-linear variable at least once.
By assumption, any non-linear variable $v$ from $\Gamma$ appears in either $e_1$ or $e_2$.
By induction, $v$ therefore appears in either the non-linear partial of $e_1$ or of $e_2$.
If $v$ appears in $E_1$, $e_1'$, $E_2$, or $e_2'$ directly, then $v$ appears in either $E$ or $e'$.
Otherwise, $v$ must either appear in the tape $\many{w}_{1}$ due to being free in $\dot e_1'$, or in $\many{w}_{2}$ due to being free in $\dot e_2'$.
Therefore, $v$ is free in $\dot e'$, because the latter only binds $\dbar v$, which cannot shadow $v$, being linear variables.
Hence we will add $v$ to the $\many w$, and it will again appear in the non-linear partial $E\parens{e, \many w}$ of $e$.
A similar argument covers the non-linear variables $\many v$ bound by the $\ttt{let}$.
Checking variable consumption for the residual is a simpler variation of the same argument.

Some notes:
The \rl{UDef} rule coordinates with the \rl{UApp} rule to transport the tape in a single variable of tuple type rather than increasing the size of every call site by the number of intermediate values of the function.
The \rl{UApp} rule relies on program ordering, because we read the type of the unzipped callee $\nonlin f$ in order to annotate the type of the call site.\footnote{%
This is trivial for \Lone\ as presented, because it does not admit recursion; but this point becomes more
subtle if recursion is permitted.
Unzipping in a recursive language must be able to construct recursive types to represent the tapes of recursive computations.
This can be accomplished by indirection: for each function $f$, allocate a fresh name $\pi_f$ for the type of the tuple carrying the tape.
Then unzipping can use that name fill in type annotations of call sites of $f$ before $f$ is unzipped, and unzipping $f$ provides the definition of $\pi_f$ (which may refer to $\pi_f$ if $f$ is recursive).
}
To enforce the invariant that non-linear expressions of \Lonestrict\ do not consume linear variables, we explicitly unzip $\drop$ expressions in rule \rl{UDrop}.
The proof of claim 4 requires induction on the program to assert that any functions called by $e$ compute the correct values, but is otherwise straightforward.
\end{proof}

\section{Transposition}
\label{sec:transposition}

\vspace{.4em}
\begin{displayquote}
To JVP, or VJP? That is the question:\\
Whether 'tis nobler in the mind to suffer\\
The slings and arrows of outrageous re-implementation,\\
Or to take up Arms against a Sea of troubles,\\
And by transposing, end them.
\end{displayquote}
\vspace{.4em}

We turn to our main focus: transposing linear functions.
Recall that our ultimate goal is to derive reverse-mode automatic differentiation from forward mode, without having to duplicate derivative rule implementations.
We do this through the notion of transposition.  Every linear function $f: \R^n \to \R^m$ has a unique transpose $f^\T: \R^m \to \R^n$ defined by consistency with respect to dot product:
\begin{equation}
    \dotp{x}{f^\T (y)} = \dotp{f(x)}{y}. \label{eq:dot-product}
\end{equation}
For a given nonlinear function $F: \R^n \to \R^m$ and input $x \in \R^n$, this is exactly the relationship between the forward derivative $dF : (y \in \R^n) \mapsto (JF|_x) y$ and the reverse derivative $dF^\T : (z \in \R^m) \mapsto z(JF|_x)$, where $JF|_x$ is the Jacobian of $F$ at $x$.

Concretely, to compute a gradient of $F:\R^n \to \R$ at $\melval{x} \in \R^n$, we differentiate $F$ with $\J$ to produce $F^\J$, unzip it into the linear and non-linear components $\nonlin {F^\J}, \lin {F^\J}$, transpose the latter to $\lin{F^\J}^\T$, and supply $1 \in \R$ as an initial gradient to the linear erasure thereof: 
\[ \nabla F|_{\many x} = \left( \letx {\many y; \many {dx}} {\nonlin{F^\J}(\many x;)} {\Le(\lin{F^\J}^\T)(\many {dx}, 1;)} \right). \]

An example transposition is given in \Cref{fig:transpose-sin}.  We use the double-dot notation $\ddot v$ as a mnemonic for cotangent variables (not second-order derivatives).
Each cotangent variable holds a partial derivative of the final result of $F$ with respect to something, which are propagated in the opposite order from the primal computation.
We assign no formal meaning to the dots themselves: $\ddot v$ is just a different variable than $\dot v$, even though we mean for it to carry the cotangent corresponding to $\dot v$.

\begin{figure}\codeexample
\begin{minipage}{0.45\textwidth}
\begin{align*}
\defxArrow
  {& \nonlin{\ttt{f}^\J}^{1,0}_{2,0}}
  {u: \R; }
  {\R, \otimes \R; } {}
  {\\
  & \letx {v: \R; } { \sin(u) } {\\
  & \letx {du: \R; } { \cos(u) } {\\
  & \letx {w: \R; } { -v } {\\
  & \letx {t: \otimes \R} {\otimes \parens{du}} {\\
  & \parens{w, t;}}}}}}\\
\defxArrow
  {& \lin{\ttt{f}^\J}^{1,1}_{0,1}}
  {t: \otimes \R; \dot u: \R}
  {; \R} {}
  {\\
  & \letx {\otimes du; } {t} {\\
  & \letx {; \dot v: \R} { du * \dot u } {\\
  & \letx {; \dot w: \R} { -\dot v } {\\
  & \parens{; \dot w}}}}}
\end{align*}
\end{minipage}
\begin{minipage}{0.053\textwidth}
\vspace{0.8in}
$\transformsto{\T}$
\end{minipage}
\begin{minipage}{0.45\textwidth}
\begin{align*}
\defxArrow
  {& \nonlin{\ttt{f}^\J}^{1,0}_{2,0}}
  {u: \R; }
  {\R, \otimes \R; } {}
  {\\
  & \letx {v: \R; } { \sin(u) } {\\
  & \letx {du: \R; } { \cos(u) } {\\
  & \letx {w: \R; } { -v } {\\
  & \letx {t: \otimes \R} {\otimes \parens{du}} {\\
  & \parens{w, t;}}}}}}\\
\defxArrow
  {& {\lin{\ttt{f}^\J}^\T}^{1,1}_{0,1}}
  {t: \otimes \R; \ddot w: \R}
  {; \R} {}
  {\\
  & \letx {\otimes du; } {t} {\\
  & \letx {; \ddot v: \R} { -\ddot w } {\\
  & \letx {; \ddot u: \R} { du * \ddot v } {\\
  & \parens{; \ddot u}}}}}
\end{align*}
\end{minipage}
\caption{Completing our running example by transposing the derivative.  On the left, we have the forward phase $\nonlin{\J(f)}$ and residual $\lin{\J(f)}$ of $f(u) = -\sin(u)$ from \Cref{fig:unzip-sin}.
On the right, we transpose just the linear residual $\lin{\J(f)}$ to obtain the reverse phase $\T(\lin{\J(f)})$, which computes the derivative of $f$ in reverse mode.}
\label{fig:transpose-sin}
\end{figure}

\subsection{Details}

\subsubsection{Input to Transposition}
We transpose a function that's well-typed in \Lonestrict\ with the $\T$ transformation\footnote{$\T$ is available in JAX as the \texttt{jax.linear\_transpose} function. Note that JAX does not distinguish between linear and non-linear expressions in its language, so the relevant typing judgements are reconstructed on the fly during transposition. The transform assumes that the given function can be typed according to the type system presented here and any failure to infer a valid type is reported to the user as an error.} given in \Cref{fig:transposition}.
Recall from \Cref{def:linear-b} that \Lonestrict\ is the sub-language of \Lone\ where all expressions must return either all-linear or all-non-linear results.
We restrict to \Lonestrict\ because general \Lone\ functions are free to interleave linear with non-linear computations, but we transpose only the strictly linear fragment.
Converting an arbitrary \Lone\ expression into \Lonestrict\ was the purpose of the unzipping transformation $\Pa$ of \Cref{sec:partial}.

\begin{figure}\footnotesize
\renewcommand{\arraystretch}{3}
\begin{tabular}{cc}
\multicolumn{2}{c}{
\framebox{$\opT{\seqTan, \seqT} {\dot e} \defunc \ddot e$}} \\
\multicolumn{2}{c}{
$\Infer{TRet}{~~}
  {\opT {\langle \many{\dot v: \dot \tau} \rangle, \langle \many{\ddot v: \typet{\dot \tau}} \rangle} {\parens {;\dbar v} } \defunc
    \parens{;\ddbar v}}$}\\
\multicolumn{2}{c}{
$\Infer{TLetLin}
  {     \seqTan_2 = \langle \many{\dot w : \dot \sigma} \rangle
  \quad \opT {\langle \many{\dot w:\dot \sigma}, \many{\dot v:\dot \tau} \rangle, \seqT} {\dot e_2} \defunc \ddot e_2
  \quad \opT {\seqTan_1, \langle \many{\ddot v:\dot \tau} \rangle} {\dot e_1} \defunc \ddot e_1
  \quad \ddbar v, \ddbar w\ \trm{fresh}}
  {\opT{\seqTan_1 \uplus \seqTan_2, \seqT} {\letx {;\many{\dot v:\dot \tau}} {\dot e_1} {\dot e_2}}
  \defunc \letx {;\many{\ddot w:\typet{\dot \sigma}}, \many{\ddot v:\typet{\dot \tau}}} {\ddot e_2} {\parens{;\ddot e_1, \ddbar w}}}$}\\
\multicolumn{2}{c}{
$\Infer{TLetNonLin}
  {     \opT {\seqTan, \seqT}{\dot e_2} \defunc \ddot e_2
  \quad \{\many v\}\ \trm{non-empty}}
  {\opT{\seqTan, \seqT} {\letx {\many {v:\tau};} {e_1} {\dot e_2}} \defunc
    \letx {\many{v:\tau};} {e_1} {\ddot e_2}}$}\\
\multicolumn{2}{c}{
$\Infer{TApp}{~~}
  {\opT {\langle \many{\dot v:\dot \tau} \rangle, \langle \many{\ddot v:\ddot \tau} \rangle} {\lin{f} \parens {\many v; \dbar v} } \defunc
    \lin{f}^\T\parens{\many v; \ddbar v} }$}\\
\multicolumn{2}{c}{
$\Infer{TNonLin}
  {e\ \trm{non-linear}}
  {\opT {\varepsilon, \varepsilon} e \defunc e}$}\\
$\Infer{TLinVar}{~~}
  {\opT {\langle \dot v:\dot \tau \rangle, \langle \ddot v:\typet{\dot \tau} \rangle} {\dot v} \defunc \ddot v}$ &
$\Infer{TLinZero}{~~}
  {\opT {\epsilon, \langle \ddot v:\typet{\dot \tau} \rangle} {0_{\dot \tau}} \defunc \drop(\ddot v)}$ \\
\multicolumn{2}{c}{
$\Infer{TLinTup}
  {     \ddbar w\ \trm{fresh}}
  {\opT {\langle \many{\dot v: \dot \tau} \rangle, \langle \ddot v:\otimes \many{\typet{\dot \tau}} \rangle} {\otimes \parens{\dbar v}} \defunc
  \letx {\otimes \ddbar w: \otimes \ddbar \tau} {\ddot v} {\parens{;\ddbar w}}}$}\\
\multicolumn{2}{c}{
$\Infer{TLinAdd}{~~}
  {\opT { \langle \dot v_1:\dot \tau, \dot v_2:\dot \tau \rangle , \langle \ddot v:\typet{\dot \tau} \rangle} {\dot v_1 + \dot v_2} \defunc
   \dup(\ddot v) }$} \\
\multicolumn{2}{c}{
$\Infer{TLinMul}{~~}
  {\opT {\langle \dot v_2:\dot \tau \rangle, \langle \ddot v_2:\typet{\dot \tau} \rangle} {v_1 * \dot v_2} \defunc
   v_1 * \ddot v_2 }$}\\
$\Infer{TDup}{~~}
  {\opT {\langle \dot v:\dot \tau \rangle, \langle \ddot v_1:\typet{\dot \tau}, \ddot v_2:\typet{\dot \tau} \rangle} {\dup(\dot v)} \defunc
   \ddot v_1 + \ddot v_2 }$&
$\Infer{TDropLin}
  {\{\dbar v\}\ \trm{are the free variables of $\dot e$}}
  {\opT {\langle \many{\dot v:\dot \tau} \rangle, \epsilon} {\drop(\dot e)} \defunc \parens{; 0_{\typet{\dot \tau_i}}}}$\\
\multicolumn{2}{c}{
\framebox{$\T(\lin{f}) \defunc \lin{f}^\T$}} \\
\multicolumn{2}{c}{
$\Infer{TDef}
  {     \opT {\langle \many{\dot v: \dot \tau} \rangle, \langle \many{\ddot v: \typet{\dot \sigma}} \rangle}{e} \defunc e'
  \quad \ddbar v\ \trm{fresh}}
  {\T (\defx {\lin f} {\many{v:\tau}; \many{\dot v: \dot \tau}} {;\dbar \sigma} e ) \defunc
    \defx {\lin f^\T} {\many{v:\tau}; \many{\ddot v: \typet{\dot \sigma}}} {;\typet{\dbar \tau}} {e'}}$}
\\
\end{tabular}
\caption{\Lonestrict\ transposition.
\Lonestrict\ is the subset of \Lone\ where all functions and expressions return only linear or only non-linear results, and non-linear expressions do not read linear variables.
Transposition leaves non-linear expressions as they are, but order-reverses linear expressions: a linear expression $\dot e$ with $n$ free linear variables and $m$ linear results becomes a linear expression $\ddot e$ with $m$ free linear variables and $n$ results computing the (linear-algebraic) transpose of $\dot e$.
}
\label{fig:transposition}
\end{figure}

\subsubsection{Transposition Environments}
To read the transposition rules, it helps to view an expression as a function from its environment to its results, and to think of its environment as ordered.  So, if $\dot e$ consumes the linear variables of an ordered type environment $\seqTan$, the transpose $\ddot e$ must produce one result corresponding to each variable in $\seqTan$.
Likewise, for every result of $\dot e$, $\ddot e$ must consume a corresponding free variable.  The transform $\T$ carries names for these (cotangent) variables in a second ordered type environment $\seqT$.

We use the notation $\seqTan$, $\seqT$ and angle brackets to emphasize that they are ordered vectors instead of the unordered mappings $\envTan$, $\envT$, and match the order of the relevant expression results.
When we compose and decompose expressions and their corresponding ordered environments, it may be necessary to insert shuffling operations to maintain this ordering invariant.
These shuffles are uninteresting and expressible using only let and multi-value return, both of which cost 0 in our model, so we leave them implicit.
The ordering is irrelevant for type-checking, so we freely use $\seqTan$ and $\seqT$ as typing environments as well.

\subsubsection{Omitted Rules}
\Cref{fig:transposition} omits the rules \rl{TUnpack} and \rl{TLinUnpack} for non-linear and linear tuple unpacking, because they are identical to the corresponding multiple-value let rules, mutatis mutandis.
There are no specific rules for transposing an expression that produces non-linear results because it transposes to itself (rule \rl{TNonLin}).
\Lonestrict\ requires non-linear expressions to have no linear sub-expressions (not even through $\drop$), so transposition need not traverse them.

\subsubsection{Benefit from Explicit $\drop$ and $\dup$}
Now we have the payoff from requiring all linear variables to be explicitly consumed: without that, $\T(\dot e)$ would have to return zeros for variables in $\seqTan$ that were not free variables of $\dot e$, and would later have to do work to add those zeros to the real cotangents.
The payoff to requiring linear variables to be consumed exactly once is that rules do not need to introduce additions to merge the results of transposed sub-expressions---those cotangents are already guaranteed to correspond to different variables, and the summations are produced by the \rl{TDup} rule.

\subsection{Formalization}

The goal of transposition is to transform a linear function $f$ into its transpose $f^\T$, as defined by pulling back the dot product (\ref{eq:dot-product}).
We now prove that it does so, and that $f^\T$ isn't more expensive than $f$.

\begin{theorem}
\textbf{Transposition is total, correct, and amortized work-preserving.}
Consider a linear \Lonestrict\ expression $\dot e$, well typed in non-linear and linear environments $\Gamma$, $\dot \Gamma$.
Fix an ordering of variables in $\dot \Gamma$ to form $\seqTan$.
Let $\melval{x}$ and $\mdelval{x}$ be well-typed values corresponding to the free non-linear and linear variables in $\dot e$, respectively.
Let $\mddelval{x}$ be values matching the types of the (linear) results from $\dot e$, and let $\inargs$ and $\outargs$ be the number of scalars present in the linear inputs $\mdelval{x}$ and (linear) results $\mddelval{x}$, respectively.
Let $\W[\dot e]$ denote the amount of work done to execute $\dot e$.
Let $\seqT$ be an ordered environment giving variable names and types for the cotangents $\mddelval{x}$ of the (linear) results of $\dot e$.  Then
\begin{enumerate}
    \item $\ddot e = \T(\seqTan, \seqT | \dot e)$ exists and is well-typed in \Lonestrict, in environments $\Gamma$; $\seqT$, and emits linear results corresponding to $\seqTan$;
    \item $\W[\ddot e] + \inargs \leq \W[\dot e] + \outargs$; and
    \item $\dotp{\mdelval{x}}{\ddot e[\melval{x}; \mddelval{x}]} = \dotp{\dot e[\melval{x}; \mdelval{x}]}{\mddelval{x}}$.
\end{enumerate}
\label{thm:transpose}
\end{theorem}

Claim 2 is the amortized work preservation criterion for transposition.
The amortization corrections $\inargs$ and $\outargs$ are the number of scalar linear inputs consumed by the original and transposed expressions, respectively.
To see intuitively why a correction is needed, consider transposing a function that takes a single argument and fans it out (with $\dup$) $n$ times.
In our cost model, this function does no work, because $\dup$ is free; but its transpose must perform $n$ summations, and thus pay $n$.
We restore work preservation by crediting an extra $(n+1) - 1$ for the $n$ excess inputs the transposition consumes.

Another intuition for the amortization argument is that producing $\outargs$ results obliges the caller to eventually pay $\outargs$ to $\drop$ or use them.  A transposition that reduces that burden on the caller should get credit it can use to amortize the additional summations it has to do internally, and vice versa.  These corrections cancel when \Lone\ expressions are composed, so transposition is exactly work-preserving internally, and the discrepancy only shows up for the end-to-end \Lone\ computation.

Claim 3 is the correctness criterion for transposition: $\ddot e$ pulls dot product in the output space of $\dot e$ back(ward) along $\dot e$ to its input space, which is exactly the mathematical definition of transposition of linear operators.\footnote{In the language of algebraic geometry, if $\dot e$ implements a pushforward of $e$, then $\ddot e$ implements a pullback of $e$.}
Note that on the left-hand side of Claim 3, the $\mdelval{x}$ are treated as extra-linguistic values and $\mddelval{x}$ are fed in as linear inputs to $\ddot e$, whereas on the right-hand side, the $\mddelval{x}$ are extra-linguistic and the $\mdelval{x}$ are fed in as linear inputs to $\dot e$.

While we state our results on transposition in terms of vector spaces over the familiar field $\R$ of real numbers, one can straightforwardly extend them to a variant of \Lone\ modeling modules over any ring $\ring$.
If $\ring$ is not commutative, one needs to introduce a second linear scaling primitive $\dot e \cdot e$, corresponding to scalar right-multiplication in $\ring$-modules.
Scalar left- and right-multiplication then transpose to each other.
Our dot-product notation also becomes non-commutative, and Claim 3 becomes two claims: $\dotp{\mdelval{x}}{\ddot e[\melval{x}; \mddelval{x}]} = \dotp{\dot e[\melval{x}; \mdelval{x}]}{\mddelval{x}}$ and $\dotp{\ddot e[\melval{x}; \mddelval{x}]}{\mdelval{x}} = \dotp{\mddelval{x}}{\dot e[\melval{x}; \mdelval{x}]}$.
Augmenting $\ring$ with sufficient structure to define differentiation, and adjusting forward differentiation to correctly compute derivatives, is left as an exercise for the enterprising reader.

\begin{proof}[Proof of Claim 1]
Well-typing of transposition follows by structural induction on the syntax of \Lonestrict.
The most interesting rule is $\rl{TLetLin}$.  The expression $\dot e$ is well-typed in $\Gamma; \seqTan$ by assumption, so we can split the incoming environment $\seqTan$ into a distinct union of the variables $\seqTan_1$ free in $\dot e_1$ and $\seqTan_2$ free in $\dot e_2$ (except for the $\dbar v$ bound in this $\ttt{let}$ form).  The results of the overall expression $\dot e$ are the same as those of the body $\dot e_2$, so we can transpose $\dot e_2$ with respect to the augmented linear environment $\seqTan_2 \uplus \langle \many{v:\tau} \rangle$ (up to a permutation) and the same expected results $\seqT$.  This gives us the transposed body $\ddot e_2$, which by induction is well-typed in $\Gamma$; $\seqT$, and produces results corresponding to $\seqTan_2 \uplus \langle \many{v:\tau} \rangle$ (up to a permutation).  We now need to split them into the $\ddbar v$ corresponding to the $\dbar v$, which will be read by $\ddot e_1$, and the $\ddbar w$ corresponding to $\seqTan_2$, which we will return from the result $\ddot e$.

We now have the arguments needed to transpose the bound expression $\dot e_1$: its linear environment $\seqTan_1$ and variables $\ddbar v$ corresponding to its results.  They line up correctly because the original expression $\dot e$ was well-typed as a $\ttt{let}$.  Transposing $\dot e_1$ gives us $\ddot e_1$, which reads the $\ddbar v$ from its environment and produces results corresponding to $\seqTan_1$.  Our result expression $\ddot e$ binds the $\ddbar v$ from $\ddot e_2$ so $\ddot e_1$ can read them, and splices the results corresponding to $\seqTan_1$ and $\seqTan_2$ into results corresponding to $\seqTan$.
The result $\ddot e$ uses all the variables from $\seqT$ exactly once because the body $\ddot e_2$ does.

The other rules proceed similarly.  One subtlety worth noting is that in \rl{TLetNonLin}, we assume that all free linear variables of $\dot e$ are in fact free in the body $\dot e_2$, and none occur in the bound expression $e_1$.
This rule is why \Lonestrict\ expressly forbids non-linear expressions from consuming linear variables, which they might otherwise do with $\drop$.
\end{proof}

\begin{proof}[Proof of Claim 2]
Work preservation needs induction on the program as well as the syntax of expressions.
The proof amounts to checking that the work corrections do, indeed, cancel.  We again consider \rl{TLetLin} first.
If we denote by $|\cdot|$ the number of linear scalars present in some set of variables, we compute
\begin{align*}
    \W[\dot e]
    & = \W[\dot e_1] + \W[\dot e_2] & \trm{cost of \ttt{let}} \\
    & \geq \W[\dot e_1] + \W[\ddot e_2] + (|\seqTan_2| + |\dbar v|) - |\seqT| & \trm{induction on $\dot e_2$} \\
    & \geq \W[\ddot e_1] + |\seqTan_1| - |\ddbar v| + \W[\ddot e_2] + (|\seqTan_2| + |\dbar v|) - |\seqT| & \trm{induction on $\dot e_1$} \\
    & = \W[\ddot e] + |\seqTan| - |\seqT|, & \trm{$\ddot v_i$ correspond to $\dot v_i$}
\end{align*}
as desired.

Handling \rl{TApp} is where we need to use induction on the program being transposed.
To wit, in this case, $\W[\dot e]$ equals the work $w$ done to evaluate (the body of) the function $\lin f$ being called, and $\W[\ddot e]$ equals the work $w^\T$ of $\lin{f}^\T$, and the corrections correspond as well.
An induction only on expression syntax would not give us the needed premise that $w \geq w^\T + \trm{correction}$, but induction on the program does, because \Lonestrict\ has no recursion.

We also find it instructive to attend to work preservation of $\rl{TDup}$.  The term $\dup(\dot v)$ itself does no work in our cost model, but the addition it transposes to does work equal to $|\dot v|$, the number of real scalars in the variable.
This is what we needed those corrections for: $\dup(\dot v)$ produces $2|\dot v|$ results, while the addition only produces $|\dot v|$ results, so the corrected work is conserved.
Similar reasoning applies to $\rl{TAdd}$ and $\rl{TLit}$.

The rule $\rl{TDrop}$ is where the inequality in work preservation may be strict: we short-circuit the subexpression $e$, just emitting zero cotangents for its free linear variables.  The resulting work $\W[0_{\T(\dot \tau_i)}]$ is zero, corrected to $|\dbar v|$.
To argue work preservation, we must invoke \Cref{lem:dead-code} to conclude that $\W[\drop(\dot e)] \geq |\dbar v|$, and thus lower-bound the corrected work of the pre-transposition expression at $|\dbar v|$ as well.
\end{proof}

\begin{proof}[Proof of Claim 3]
Correctness is also an induction on the program and the syntax of expressions.
To give a flavor for the leaf-node behavior of transposition, let's prove correctness of \rl{TLinAdd} and \rl{TDup}.
For type-compatible $\dot v_1, \dot v_2$, and $\ddot v$,
$$
  \dotp{(\dot v_1, \dot v_2)}{\dup(\ddot v)} 
  = \dotp{\dot v_1}{\ddot v} + \dotp{\dot v_2}{\ddot v} 
  = \dotp{\dot v_1 + \dot v_2}{\ddot v}.
$$
Read from left to right, this is exactly the correctness equation for \rl{TDup}, and from right to left for \rl{TLinAdd}.
Note that this simple proof critically depends on the distributivity of multiplication over addition, also known as the linear factoring rule, which is crucial for the efficiency of reverse-mode autodiff (see e.g. \citet{brunel2020backpropagation}).

Transposition on interior AST nodes is exemplified by \rl{TLetLin}.
To begin the calculation, we stipulate the disjoint partition of $\mdelval{x}$ into $\mdelval{x}_{1}$ corresponding to $\seqTan_1$ and $\mdelval{x}_{2}$ corresponding to $\seqTan_2$; and a similar (but not necessarily disjoint) split of $\melval{x}$ into $\melval{x}_{1}$ and $\melval{x}_{2}$.
To simplify notation, we name the intermediate quantities $\mdelval{z}$, $\mddelval{z}$, and $\mddelval{x}_{2}$, as, respectively, the forward evaluation of the bound expression $\dot e_1$, the corresponding backward evaluation of the transposed body $\ddot e_2$, and the residual results of $\ddot e_2$ which correspond to the free linear values $\mdelval{x}_{2}$:
\begin{align*}
    \mdelval{z} & = \dot e_1[\melval{x}_{1}; \mdelval{x}_{1}] \\
    \mddelval{z} \uplus \mddelval{x}_{2} & = \ddot e_2[\melval{x}_{2}; \mddelval{x}].
\end{align*}

Then we conclude the correctness proof for \rl{TLetLin} by calculation:
\begin{align*}
    \dotp{\dot e[\melval{x}; \mdelval{x}]}{\mddelval{x}}
    & = \dotp{\dot e_2[\melval{x}_{2}; \mdelval{z} \uplus \mdelval{x}_{2}]}{\mddelval{x}}
      & \trm{1-step evaluation of $\dot e$} \\
    & = \dotp{\mdelval{z} \uplus \mdelval{x}_{2}}{\ddot e_2[\melval{x}_{2}; \mddelval{x}]}
      & \trm{induction on $\dot e_2$} \\
    & = \dotp{\mdelval{z}}{\mddelval{z}} + \dotp{\mdelval{x}_{2}}{\mddelval{x}_{2}}
      & \trm{partitioning dot product} \\
    & = \dotp{\mdelval{x}_{1}}{\ddot e_1[\melval{x}_{1}; \mddelval{z}]} + \dotp{\mdelval{x}_{2}}{\mddelval{x}_{2}}
      & \trm{induction on $\dot e_1$} \\
    & = \dotp{\mdelval{x}}{\ddot e_1[\melval{x}_{1}; \mddelval{z}] \uplus \mddelval{x}_{2}}
      & \trm{unpartitioning dot product} \\
    & = \dotp{\mdelval{x}}{\ddot e[\melval{x}; \mddelval{x}]}.
      & \trm{1-step (un)evaluation of $\ddot e$}
\end{align*}

Other syntactic forms proceed similarly, and more simply.  \rl{TApp} again requires induction on the whole program being transposed.
\end{proof}

\begin{corollary}
If we transpose an expression $\dot e$ twice, the result is equivalent to $\dot e$ and does at most the same amount of work. 
\end{corollary}

\section{Code size}

When implementing practical compilers, the size of the (intermediate) program representation is also an important consideration alongside its ultimate runtime, because it limits the runtime and memory cost of the compiler itself.
We thus emphasize that our transformations increase code size only linearly, because they turn function definitions and function calls into new definitions and calls, without duplicating or inlining function bodies.

We can argue this formally by modeling the size $|P|$ of a \Lone\ program $P$ simply as ``every node in the grammar from \Cref{fig:linear-a} has size 1''.  We then deduce
\begin{theorem}\textbf{Erasure, differentiation, unzipping, and transposition are size-preserving.}
There exist constants $c_\Le$, $c_\T$, $c_\Pa$, $c_\J$ such that for any \Lone\ program $P$,
\begin{align*}
    |\Le(P)| & \leq c_\Le |P|, \\
    |\J(P)| & \leq c_\J |P|, \\
    |\Pa(P)| & \leq c_\Pa |P|,\ \trm{and} \\
    |\T(P)| & \leq c_\T |P|, \\
\end{align*}
where $\Le$, $\J$, $\Pa$, and $\T$ are the erasure, differentiation, unzipping, and transposition transformations defined in Sections~\ref{sec:linearity-erasure},~\ref{sec:jvp},~\ref{sec:partial}, and~\ref{sec:transposition}, respectively.
\label{thm:size-preservation}
\end{theorem}

\begin{proof}
Mostly by inspection, but there is one subtlety.  In rule \rl{UApp} (\Cref{sec:partial}), the type $\sigma$ of the tape of the applied function $f$ appears in the emitted let binding.
In \Lone\ as presented, $\sigma$ will have size proportional to the number of non-linear intermediates in $f$, which may be asymptotically larger than the call site being transformed.
This is, however, easy to fix: just augment the type system with synonyms, and annotate that binding with a synonym for the tape instead.
When function calls are nested, this amounts to retaining structure sharing in the tape types.

The same phenomenon is also why linearity erasure $\Le$ (\Cref{sec:linearity-erasure}) has to introduce dedicated functions for vector space operations on arbitrary types, instead of inlining them.
\end{proof}

The constant $c_\J$ depends on the set of primitives and their derivatives defined for \Lone, but $c_\Le$, $c_\Pa$ and $c_\T$ are absolute, given by the size of the expressions constructed by our rules for $\Le$, $\Pa$ and $\T$, respectively.
It is necessary to count variable references in program size (even though we model them with zero runtime cost), because for instance $\J$ introduces calls to $\dup$ for repeated references to (non-linear) variables.

\section{Related work}

Most of the material presented here can be found in one form or another in existing work, but the composition of individual pieces is novel.
The AD architecture we describe was sketched out briefly in
\citet{frostig2021decomposing} and again in
\citet{paszke2021dex}, and is in use in the
JAX and Dex projects~\citep{frostig2018compiling, jax2018github}.
Transposition as a program transformation has been previously explored by \citet{piponi2009transposition}.
Unzipping is an instance of partial evaluation---a technique widely known and explored \citep{futamura1983partialeval,jones1993partialeval}---with the difference that we are interested not only in the residual program, but also in the program formed by operations that are considered ``statically computable''.
Automatic differentiation itself has seen rapid growth in recent years, especially for functional languages, which we attempt to briefly summarize below.

Forward-mode differentiation is significantly simpler to embed in existing programming languages, for example by the means of  templates \citep{Piponi2004ad}, operator overloading \citep{walther2003ooad} and Haskell-style typeclasses or ML-style functors \citep{karczmarczuk1998ad,elliott2009beautiful-ad}. The true difficulty, both conceptually and implementation-wise, arises for reverse mode.

Initially, reverse-mode autodiff development was focused on imperative languages. In particular, the implicit reuse of values in the original computation was consistently translated to destructive mutation.
This made higher-order differentiation more difficult, and meant that it generally did not have a natural non-monadic translation into pure functional languages.
This challenge was overcome by \citet{pearlmutter2008lambda-tub}, with the introduction of reverse mode as an augmentation of the original program with ``backpropagator'' closures replacing the tape.
As shown by \citet{wang2019shiftreset} a similar transformation can be performed by expressing the non-local control using delimited continuations, manifested by the shift and reset operators.
Furthermore, \citet{brunel2020backpropagation} use a similar method to perform provably correct reverse-mode differentiation over lambda calculus extended with linear negation.
Especially interestingly, they already make the connection between linear logic and algebraic linearity that underlies the work presented here.
Finally, \citet{krawiec2021provably} outline a provably correct and efficient implementation of AD in a higher-order language, at a cost of requiring an extra-linguistic monadic translation.

The purpose of the monadic translation of \citet{krawiec2021provably} is to stage out a program, expressing the linearization of primal evaluation at a given point, as the primal is running.
But, since control flow is dependent only on the non-linear values, the staged program never uses control flow or higher-order functions itself (linearization of only the branches taken gets inlined). In fact, it uses a language almost identical to the linear part of \Lone.

That approach can be factorized through the lens taken in our work.
\citet{krawiec2021provably} describes reverse-mode AD in two phases.
The first phase simultaneously: traces the metalanguage into a representation without control flow and higher-order functions (our \Lone); performs forward mode AD (our $\J$); and partially evaluates (like our unzipping $\Pa$, but specializing the linear function to exact values, instead of separating non-linear compute).
In the second phase, transposition is implemented as an evaluator function that traverses the staged linear program.

Complementing the previous implementation-focused work, multiple explorations of autodiff semantics have been developed.
Many of those give up the asymptotic efficiency guarantees, which are critical for practical implementations of autodiff, but they do help shed light onto the underlying theory.
\citet{elliott2018simpleessence} outlines a minimal autodiff system for a language following category theoretical foundations.
Similarly to our work he also extensively reasons about the linearity of derivatives, including the introduction of a linear function type.
\citet{abadi2020simple} describe consistent operational and denotational semantics for autodiff in a first-order language, while \citet{huot2020correctness} extend it to a higher-order language using diffeology semantics.
\citet{mazza2021pcfad} outline semantics for a higher-order language with conditionals that are consistent with the mathematical notion of differentiation everywhere except for a measure 0 set.

\section{Discussion and future work}
\label{sec:future-work}

\textbf{Extensibility.}
Decomposing reverse-mode AD can lead to fewer AD rules overall. Consider a language with $N$ primitives, $H$ of which are higher-order (think control flow) and $L$ are linear.
If the two AD modes are separate, then a system implementer has to provide $N + N$ rules---one for each primitive for each mode.
However, in our case, $N + H + L$ suffice: one still implements forward-mode for all primitives, but unzipping acts over all first-order primitives uniformly, and transposition only deals with linear primitives.
Many first-order array languages support tens or even hundreds of primitives, most of which are both first-order and non-linear, leading to significant savings.

As a corollary, with our approach, an AD system user who wants to customize AD behavior on a subprogram only needs to provide the custom forward-mode implementation, as the customized reverse-mode implementation can be derived for them automatically.

\textbf{Control flow.}
Data-dependent control is the most glaring omission from \Lone.
While \Lone\ remains useful as a model of the straightline computation results from tracing at a given primal point, adding ``if'' or ``case'' (and loops or recursive functions) would broaden it to a model of the whole surrounding numerical language.

Control would be a substantial extension: control flow means that different branches may compute different intermediate values, so one must also extend the type system to allow sum types (which were not necessary in \Lone).
Having sum types, one must then concern oneself with defining the tangent type of a sum.
It would be most satisfying to use a dependent type system, so as to preserve the invariant that only non-linear variables may have sum types, and linear computations operate only on products of $\R$.
\emph{Which} product, however, can depend on a previous non-linear computation.
Such a type system would also capture the constraint that the value tangent to a non-linear sum must inhabit the same variant.

\textbf{Multilinearity.}
\Lone\ does not model multilinearity.  To wit, multiplication is actually linear separately in both arguments, whereas its model in our system treats the first argument non-linearly.  More generally, a function like $f(x, y, z) \mapsto z (x + y)$ can be viewed as linear in $z$, or as linear in $(x, y)$ jointly.  Our linear languages force the user to choose one, or to duplicate the function if both interpretations are necessary; whereas it would be more satisfying to be able to encode all of those linearity properties in a type system at once.

Multilinearity seems to be an important ingredient in two other directions of generalization.  First, a program containing independent derivatives will have components that are independently linear; which can be important to track, especially if those derivatives interact (e.g., if they are nested).  This is important to address in order to give a satisfying, linearity-aware typing to in-language automatic differentiation (as opposed to the extra-linguistic transformations we dealt with here).

Second, multilinearity is part of extending these ideas to higher-order functions.  When asking what it means for a function-returning function such as $f : \R \to (\R \to \R)$ to be linear, a natural answer seems to be that the final result should be linear in both inputs separately, i.e., multilinear.  However, any system covering higher-order functions should be able to give a linear type to curried addition as well, where the final result is linear in both arguments jointly.

\textbf{Partial evaluation.}
We have observed that our unzipping transformation $\Pa$ is an instance of partial evaluation.  Unzipping has a time-space tradeoff: for each non-linear intermediate, $\Pa$ can either store the result on the tape (to minimize work), or emit code in the linear residual that recomputes it (to reduce memory pressure by the emitted program).  The AD community has devoted considerable attention to this tradeoff under the heading of checkpointing AD; perhaps there are ideas there that could be ported to general partial evaluation in memory-limited settings.

\textbf{Array size polymorphism.}
One of the limitations of \Lone\ as defined is that the size of any tuple is syntactically apparent, so it is not possible to write size-polymorphic (i.e., ``array processing'') functions in \Lone.
Array-size-polymorphic functions are of course critical in practice, but are mostly orthogonal to the transformations we introduced here.
We leave a full formalization to future work, with the comment that JAX resolves array size polymorphism at trace time (in our terms, during metaprogramming of \Lone), whereas Dex handles it by extending the internal representation with constructs that deal with it directly.  %

\textbf{Sparse matrix algorithms.}
In a different direction, indexed linear functions are a representation for matrices.  
Indeed, they are a maximally expressive representation for matrices, that can encode any structured (or unstructured) sparsity pattern whatever.
This statement is not very deep: the sine qua non of a sparse matrix representation is to be able to compute matrix-vector products $M x$; the function that closes over the matrix $M$ and computes that product with a vector $x$ it accepts as argument is algebraically linear; and it seems obvious that any such function can be written to type-check as linear in $x$ in some reasonable variant of \Lone\ (perhaps extended with control flow as needed).

From this lens, our contribution is a universal sparsity-preserving transposition operation for sparse matrix representations (as well as a formalization of the intuition that an automatic derivative is a sparse representation of the Jacobian of the differentiated function).
Are there any other sparse-matrix operations that can be rendered universal as code transformations on \Lone?
Can sparse-matrix algorithms for, say, matrix multiplication be fruitfully recovered from a known code transformation (e.g., simplification or partial evaluation) applied to the composition of the corresponding indexed linear functions?
Can new sparse-matrix algorithms or representations be derived more simply or robustly from this direction?

\section{Conclusion}

We presented a decomposition of reverse-mode automatic differentiation into three distinct program transformations: forward-mode differentiation, unzipping, and transposition.
The interface between these transformations is a (substructural) linear type system for checking (algebraic) linearity.
Automatic derivatives type-check as linear, unzipping separates the linear and non-linear parts of a computation, and transposition runs the linear part backward to efficiently compute its transpose.

This decomposition clarifies and simplifies automatic differentiation systems, by separating the mind-bending direction-reversal needed for reverse-mode AD from the actual differentiating.
This decomposition also sheds light on checkpointing strategies---the decision of whether to save an intermediate value or recompute it is entirely the province of the unzipping transformation, which amounts to performing partial evaluation in a way that trades off run-time computation for lower memory usage.
Transposition handles any set of choices mechanically.

\section{Acknowledgements}

The authors would like to thank Pavel Sountsov for the observation that \Lone\ does not handle shape polymorphism,
Colin Carroll for the use case of simplifying custom gradients,
and Gordon Plotkin for introducing us to linearity, especially
the connection between linear algebra and linear types.
We also thank our anonymous reviewers for their remarkably detailed and insightful feedback.

\bibliographystyle{ACM-Reference-Format}
\bibliography{ref}

\end{document}